\documentclass[11pt]{article}
\usepackage[margin=0.8in]{geometry}
\usepackage{booktabs} % For formal tables
\usepackage[ruled,linesnumbered]{algorithm2e} % For algorithms

\SetAlFnt{\small}
\SetAlCapFnt{\small}
\SetAlCapNameFnt{\small}
\SetAlCapHSkip{0pt}
\IncMargin{-\parindent}

\makeatletter
\AddToHook{cmd/appendix/before}{\def\cref@section@alias{appendix}}
\makeatother

\usepackage{enumitem}

\usepackage{url}              % simple URL typesetting   
\usepackage{amsfonts,amssymb,mathtools,bbm,verbatim}  
\usepackage{amsthm}
\usepackage[authoryear]{natbib} %  (Author, Year)

\newtheorem{claim}{Claim} % added by Paritosh
\newtheorem{lemma}{Lemma}
\newtheorem{theorem}{Theorem}
\newtheorem{definition}{Definition}

\usepackage{thm-restate}
\usepackage{mathtools}
\usepackage{bbm}
\usepackage[capitalize]{cleveref}

\newcommand{\items}{[m]}
\newcommand{\agents}{[n]}
\newcommand{\alloc}{\mathcal{A}}

\newcommand{\vtrue}{v^*}

\newcommand{\vestimate}{\hat{v}}

\newcommand{\I}{\mathbbm{1}}

\newcommand{\D}{\mathcal{D}}
\newtheorem{application}{Application}

\crefname{algoLine}{line}{lines}
\Crefname{algoLine}{Line}{Lines}

\def\sF{{\psi}}

\def\P{{\mathbb{P}}}
\def\E{{\mathbb{E}}}

\def\var{{\mathrm{Var}}}

\def\R{{\mathbb{R}}}

\def\1{{\mathbf{1}}}
\def\0{{\mathbf{0}}}
\def\I{{\mathbbm{1}}}

\newcommand{\T}{\mathrm{T}} % transpose
\DeclareMathOperator*{\argmax}{arg\,max}
\DeclareMathOperator*{\argmin}{arg\,min}

\newcommand\numberthis{\addtocounter{equation}{1}\tag{\theequation}}

\title{Fair Division Under Inaccurate Preferences}

\author{Trung Dang\thanks{The University of Texas at Austin. {\tt dddtrung@cs.utexas.edu}} \and Daniel Halpern\thanks{Google Research. {\tt dhalpern@google.com}} \and Anuran Makur\thanks{Purdue University. {\tt amakur@purdue.edu}} \and Alexandros Psomas\thanks{Purdue University and Google Research. {\tt apsomas@purdue.edu}} \ \and Japneet Singh\thanks{Purdue University. {\tt japneetsingh776@gmail.com}} \and Paritosh Verma\thanks{Purdue University. {\tt paritoshverma97@gmail.com}}}

\date{}

\begin{document}

\maketitle

% Abstract. Note that this must come before \maketitle.
\begin{abstract}
    The fair allocation of scarce resources is a central problem in mathematics, computer science, operations research, and economics. While much of the fair-division literature assumes that individuals have underlying cardinal preferences, eliciting exact numerical values is often cognitively burdensome and prone to inaccuracies. A growing body of work in fair division addresses this challenge by assuming access only to ordinal preferences. However, the restricted expressiveness of ordinal preferences makes it challenging to quantify and optimize cardinal fairness objectives such as envy.

    In this paper, we explore the broad landscape of fair division of indivisible items given inaccurate cardinal preferences, with a focus on minimizing envy. We consider various settings based on whether the true preferences of the agents are stochastic or worst-case, and whether the inaccuracies, modeled as additive noise, are stochastic or worst-case. When the true preferences are stochastic, we show that envy-free allocations can be computed with high probability; this is achieved both in the setting with stochastic and worst-case noise. This generalizes a notable result in stochastic fair division, which establishes a similar guarantee, albeit in the absence of any noise. When the true preferences are worst-case, and the noise is bounded, we analyze the maximum envy achieved by the well-known Round-Robin algorithm. This bound is shown to be tight for deterministic algorithms, and applications of this bound involving known results from statistical estimation are provided. Lastly, we consider a setting with worst-case preferences and noise, where the true preferences for each item are revealed upon its allocation. Here, we give an efficient online algorithm that guarantees logarithmic maximum envy with high probability, for bounded noise. This result builds upon and generalizes a known result from algorithmic discrepancy to a setting with noisy input data. The broad range of tools used in our results highlights the theoretical relevance of this natural fair division problem.
\end{abstract}

% Title page for title and abstract only.
% \begin{titlepage}

% \maketitle

% Optionally include a table of contents
% \vspace{1cm}
% \setcounter{tocdepth}{2} % adjust to 1 if desired
% \tableofcontents

% \end{titlepage}

\section{Introduction} \label{sec: Introduction}
% 6-8 pages

Fair division studies the fundamental problem of fairly allocating resources to agents with heterogeneous preferences over these resources. Mathematical formulations date back to~\cite{steinhaus1948problem}
and~\cite{dubins1961cut}, though some of the field's core questions (and even key algorithms) can be traced back to the Bible and ancient Greece. Most of the existing literature assumes that agents' preferences are represented by valuation functions, which the agents report to the designer, or the designer can access through exact valuation queries.

While we have vastly deepened our understanding of when fairness is --- and when it is not --- attainable under these models, when taking this approach to practice, an especially unappealing aspect is that participants must report \emph{exact} numerical values to the algorithm. Motivated by the reality that eliciting cardinal valuations is prone to inaccuracies as well as cognitively burdensome for the participants, a growing body of work studies the problem of a designer that only has access to \emph{ordinal information}.
This idea originates from the seminal work of~\cite{procaccia2006distortion} that defines the concept of distortion: the worst-case loss with respect to some aggregate
cardinal objective (e.g., utilitarian social welfare) due to limited expressiveness; see~\cite{anshelevich2021DistortionIS} for a survey.

Most works in this area take a worst-case approach; after the algorithm has made its decisions, the underlying cardinal valuations are chosen adversarially (but, of course, consistently with the ordinal information available to the algorithm). This approach gives robust, albeit pessimistic guarantees. While we agree on the assumption that participants being able to submit \emph{precise} numerical values is unrealistic in practice, we also seek a model which takes into account the fact that participants \emph{can} provide useful, but perhaps imperfect, cardinal signals.

%Notably, this worst-case viewpoint stands in sharp contrast to the vast literature on discrete choice models, which, like our problem here, also assumes underlying cardinal utilities and access to ordinal information.
%Standard discrete choice models assume that the ordinal observations are \emph{stochastic}. For example, the Bradley-Terry-Luce (BTL) model (sometimes referred to as the multinomial logit model), arguably the most popular discrete choice model, attempts to estimate the (cardinal) utility parameters of the BTL model based on (ordinal) comparison data. However, such attempts to recover cardinal information through ordinal comparison data are inevitably prone to inefficiency, since the former is more informative than the latter. Working with ordinal preference data, in the context of fair division, also prevents us from quantifying and optimizing prominent cardinal notions of fairness like envy.

Hence, in this paper, we consider a middle ground between the ``perfectly accurate cardinal'' and ``fully ordinal'' models. We initiate the study of fair division given access to inaccurate cardinal preferences. Our goal is to design algorithms that minimize the maximum envy between agents (with respect to true preferences), using only inaccurate preferences as input. 
%An agent is said to envy another agent if the value of the items that the second agent has, as per the first agent, is higher than the value that the first agent has for its own items. Achieving envy-freeness, or the general pursuit of minimizing envy between agents, is one of the most well-studied objectives in fair division~\cite{varian1973equity,lipton2004approximately}. 
We investigate various formal realizations of this general fair division problem, such as different 
ways to model inaccurate cardinal preferences as well as different assumptions about agents' true preferences.

\subsection{Our Contributions}

We introduce the problem of fair division with imperfect information. There are $n$ agents with additive preferences, and $m$ indivisible goods/items. Given noisy estimates $\vestimate_{i,j}$ of the true values $\vtrue_{i,j}$, for each agent $i$ and item $j$ as input, our goal is to allocate the items in a way that minimizes the envy, with respect to the true values $\vtrue_{i,j}$. We model inaccuracies as additive noise on the true values. Our results are broadly divided into two categories depending on whether the true values of the agents are stochastic or worst-case.\\

\noindent
\textbf{Stochastic true valuations.} These results focus on scenarios in which the true valuations $\vtrue_{i,j}$ are stochastic. We consider two subsettings, depending on whether the added noise is stochastic or worst-case. In both cases, we establish the existence of envy-free allocations, with high probability.

\begin{itemize}[leftmargin=*]
\item First, we consider the case of stochastic true valuations and stochastic noise, i.e., $\vtrue_{i,j} \sim \mathcal{D}$, $\eta_{i,j} \sim \mathcal{D}'$, where we get to observe inaccurate/noisy values $\vestimate_{i,j} = \vtrue_{i,j} + \eta_{i,j}$; here $\mathcal{D}$, $\mathcal{D}'$ are two unknown, fixed, real-valued distributions. In \Cref{theorem:observed-value-maximization}, we show that the simple procedure of maximizing welfare with respect to the observed values produces envy-free allocations with high probability. This generalizes a result of \cite{dickerson2014computational} to the setting with inaccurate preferences.\footnote{\cite{dickerson2014computational} show that envy-free allocations exist with high probability in a setting where $\eta_{i,j} = 0$ for all $i,j$.} To prove our result, we establish a strict statistical correlation inequality (\Cref{lemma:strict-fkg}) that might be of independent interest. We show that for any two real-valued, nondecreasing functions $f$, $g$ and any real-valued random variable $X$, we have $\E[f(X)\cdot g(X)] > \E[f(X)] \cdot \E[g(X)]$ if and only if $\var(f(X)) > 0$ and $\var(g(X)) > 0$. This is used to establish the inequality $\mathbb{E}[\vtrue_{i,j} \mid i \in \argmax_{k \in [n]}\{\vestimate_{k,j}\}] > \mathbb{E}[\vtrue_{i,j} \mid i \notin \argmax_{k \in [n]}\{\vestimate_{k,j}\}]$ (\Cref{lemma:noisy-order-statistic}), which, in turn, allows us to argue that agents have higher true values ($\vtrue_{i,j}$) for items they are allocated (compared to the ones they are not allocated), where the allocation decisions are based on the noisy values ($\vestimate_{i,j}$). Note that the weaker inequality, $\mathbb{E}[\vtrue_{i,j} \mid i \in \argmax_{k \in [n]}\{\vestimate_{k,j}\}] \geq \mathbb{E}[\vtrue_{i,j} \mid i \notin \argmax_{k \in [n]}\{\vestimate_{k,j}\}]$, follows directly from Chebyshev's association inequality; unfortunately though, it is insufficient for our result, which necessarily requires a strict inequality.
\item For stochastic true valuations and worst-case noise, we consider a general setting where the value vector for each item $j$ is drawn i.i.d.\@ from a bounded non-negative vector-valued distribution $\D$ over $[0,b]^n$; that is, agents' valuations for a specific item are arbitrarily correlated, but values across items are independent. We prove that, if $\D$ satisfies a mild technical condition (essentially that no two agents are identical up to scaling) then as long as the noise is bounded by a constant, there exists an algorithm that, given noisy estimates, outputs an allocation that is envy-free with high probability (\Cref{theorem:stochastic-value-3-worst-case-noise}). To prove \Cref{theorem:stochastic-value-3-worst-case-noise}, we establish the existence of strongly envy-free allocations, with high probability, in the underlying noiseless setting, by adapting an existence result from the literature on cake-cutting~\cite{barbanel2005geometry}.
\end{itemize}

% We consider four regimes depending on whether the true values are worst-case or stochastic, and whether the noisy estimates are worst-case or stochastic.

\noindent
\textbf{Worst-case true valuations.} These results focus on scenarios where the underlying true valuations $\vtrue_{i,j}$s are fixed (i.e., worst-case).

\begin{itemize}[leftmargin=*]
\item Our first set of results concerns the case of worst-case true valuations and worst-case noise, i.e., we assume that the unknown true values $\vtrue_{i,j} \in [0,b]$ are fixed, and we have estimates $\vestimate_{i,j}$ that satisfy $|\vestimate_{i,j} - \vtrue_{i,j}| \leq \varepsilon$ for some (possibly unknown) $\varepsilon \geq 0$. Here, we establish a tight bound (\Cref{thm: worst-case valuations and worst-case error upper bound,thm: worst-case valuations and worst-case error lower bound}) on the best possible envy guarantee that can be achieved by any deterministic algorithm that only has access to $\vestimate_{i,j}$. In particular, we show that the standard Round-Robin algorithm\footnote{The Round-Robin algorithm operates in rounds, where in each round, agents follow an order $1,2,\ldots, n$, and pick, from the set of unallocated/available items, an item having the highest value. Here we assume that these picks are made based on the noisy estimates $\vestimate_{i,j}$.} achieves an optimal guarantee. %While the proofs of both the upper and lower bounds, which establish this tight guarantee, are based on standard arguments, we note that our choice of using round robin is crucially based on the fact that it ensures that all the agents get roughly equal \emph{number} of items. As we discuss in \Cref{section:bringing-together}, this property enables us to use the tools of BTL parameter estimation for fair division. %, making round robin essentially the unique deterministic algorithm suitable for our two-step framework.

\item We proceed to give two applications of our Round-Robin result. 

In the first application (\Cref{application:combined_round_robin}), we present an end-to-end framework where we use \emph{ordinal information} to estimate values, and then use the estimated values to make allocation decisions. Specifically, we assume access to stochastic ordinal observations; we choose arguably the most popular discrete choice model, the Bradley-Terry-Luce
(BTL) model (sometimes referred to as the multinomial logit model). In this model, the probability of an agent ranking one alternative higher than another is given by the logistic function applied to the difference between the valuations of the two options~\cite{Zermelo1929,BradleyTerry1952,Luce1959,Plackett1975,McFadden1973}. Given such ordinal information, we first produce estimates $\vestimate_{i,j}$ for the true preferences. Then these estimates are plugged into the Round-Robin algorithm to get an allocation with a bounded maximum envy. The fact that Round-Robin outputs allocations that are \emph{balanced}\footnote{An allocation $\alloc = (A_1, A_2, \ldots, A_n)$ is called balanced if all agents get roughly equal number of items, i.e., $|A_i| - |A_j| \leq 1$ for all $i,j \in [n]$.} in addition to being fair is a key property that enables us to combine BTL and Round-Robin in an end-to-end framework. 

The second application (\Cref{application:worst-case-true-val-stochastic-noise}) deals with stochastic noise with an \emph{unbiased} error. Specifically, we assume that $\vestimate_{i, j} \coloneqq \vtrue_{i, j} + \varepsilon_{i, j}$, where $\varepsilon_{i,j} \coloneqq \sigma_{i,j} \cdot z_{i,j}$ where $z_{i, j}$s are drawn i.i.d. from a non-negative (possibly unknown) distribution $\D$ and $\sigma_{i,j}$s are drawn uniformly at random from the set $\{ -1 , 1 \}$. We focus on small tail distributions $\D$, and in particular, MHR distributions.\footnote{A distribution $\D$ has monotone hazard rate (MHR) iff $\frac{1 - F(x)}{f(x)}$ is non-increasing, where $F(x)$ and $f(x)$ are its cumulative distribution function and probability density function, respectively.} We show that, given $\vestimate_{i,j}$s, allocations with at most a constant maximum envy can be computed efficiently. Towards this, we employ known results about the order statistics of MHR distributions and a recent result concerning the concentration of order statistics of MHR distributions. \Cref{application:worst-case-true-val-stochastic-noise} is obtained by combining the bound on the largest noise with our aforementioned Round-Robin result.
\item Finally, after exploring various aspects of worst-case valuations and bounded noise (the tight Round-Robin guarantee and its applications), we proceed to study a different, more permissive noise model. In particular, we focus on a scenario where, just as before, we have to allocate items to agents based on the noisy estimates $\vestimate_{i,j}$; however, this time, once any item $j$ is allocated, then the true valuation vector $\vtrue_j$ corresponding to this item is revealed.\footnote{The revelation of information (preference) that governs the objective (envy), following a decision of the algorithm (allocation of an item) is also a key theme in many well-studied problems like the online convex optimization (OCO), the experts problem, etc.} This immediate revelation of the true preference gives more power to the algorithm designer. In this model, we show that there is a simple, online algorithm that achieves a maximum envy bound of $\mathcal{O}\left(\sqrt{n}(\log{\frac{mn}{\delta}} + m\varepsilon)\right)$ with probability at least $1-\delta$ (\Cref{theorem:envy-minimization-discrepancy}), where the noise is bounded by $\varepsilon$, i.e., $|\vtrue_{i,j} - \vestimate_{i,j}| \leq \varepsilon$. The algorithm achieving this bound (\Cref{algo:BalanceUnderNoise}) is inspired by (and is a modified version of) an online discrepancy minimization algorithm given by~\cite{alweiss2021discrepancy}.\footnote{\cite{alweiss2021discrepancy} study the problem of online discrepancy minimization in the absence of any noise.} Notably, if the noise satisfies the upper bound $\varepsilon \leq \frac{1}{m}\log{\frac{mn}{\delta}}$, then we get an envy as well as a discrepancy bound of $\mathcal{O}\left(\sqrt{n}\cdot \log{\frac{mn}{\delta}}\right)$ that grows logarithmically in $m$. This matches the result of \cite{alweiss2021discrepancy} in the noiseless setting, and is obtained by a comparatively simpler induction-based analysis. To establish this result, we first bound discrepancy in the online vector balancing problem (\Cref{corollary:infinity-norm-bound}), and then bound multicolored discrepancy in an online setting with noisy vectors (\Cref{theorem:multicolored-discrepancy-bound}). A known reduction from envy minimization to multicolored discrepancy minimization then gives us our final envy minimization result (\Cref{theorem:envy-minimization-discrepancy}).
\end{itemize}

\subsection{Additional Related Work}
Most literature on fair division assumes that agents' preferences are captured by valuation functions. An (often implicit) assumption in this literature is that the agents have the ability to report \emph{exact} numerical values to the designer. In practice, eliciting cardinal valuations is error-prone, as well as cognitively burdensome. A growing body of work focuses on settings where access to only \emph{ordinal information} is available.  A central notion is distortion~\cite{procaccia2006distortion}: the worst-case loss with respect to some aggregate
cardinal objective resulting from the access to limited ordinal information. Bounds on distortion have been studied for problems, including fair division~\cite{amanatidis2016truthful,halpern2021},
matching~\cite{anshelevich2016blind,abramowitz2018utilitarians,anshelevich2019tradeoffs,ebadianevery},
voting~\cite{boutilier2012optimal,anshelevich2018approximating} %,caragiannis2017subset,goel2017metric,mandal2020optimal,munagala2019improved,kempe2020communication,gkatzelis2020resolving,kizilkaya2022plurality,charikar2024breaking,voudouris2025metric,gonczarowski2024distortion,gkatzelis2023best}, 
and participatory budgeting~\cite{benade2021preference} (see~\cite{anshelevich2021DistortionIS} for a survey).

Closer in spirit to our setting are works that elicit ordinal feedback (often via pairwise comparisons) while allowing the underlying cardinal valuations to be chosen adversarially, e.g.,~\cite{benade2025dynamic,amanatidis2016truthful,halpern2021,ebadian2022efficient,amanatidis2022few}. Our results for stochastic true valuations also relate to the well-studied area of stochastic fair division~\cite{dickerson2014computational,bai2021envy,benade2023existence}. Finally,~\cite{bai2022fair} study a \emph{smoothed} model, where a worst-case utility profile is slightly perturbed at random. In~\cite{bai2022fair}, the perturbed profile constitutes agents' true valuations; in contrast, we keep true valuations fixed and the perturbed profile is the observed instance.
\section{Preliminaries}\label{preliminaries:fair-division}

We consider the problem of dividing a set of $m$ indivisible items among a set of $n \ge 2$ agents. We will represent the set of agents by $[n]$ and the set of items by $[m]$.  Each agent $i \in \agents$ has a valuation function $\vtrue_i : 2^{[m]} \rightarrow \mathbb{R}$, which maps every subset of items $S \subseteq \items$ to $\vtrue_i(S)$, the agent's value for $S$. In this paper, we focus on the case of \emph{additive} valuation functions. That is, each agent $i \in \agents$ has a value $\vtrue_{i,j} \in \mathbb{R}$ for each item $j \in \items$, and the value of agent $i$ for a subset of items $S \subseteq \items$ is $\vtrue_i(S) = \sum_{j \in S} \vtrue_{i,j}$.

\paragraph{Allocations.} An allocation $\alloc = (A_1, \ldots, A_n)$ is an $n$-partition of $\items$, where $A_i$ denotes the set of items allocated to agent $i \in \agents$. That is, in any allocation $\alloc$, all items are allocated, $\bigcup_{i \in \agents} A_i = \items$, and each item is allocated to only one agent, $A_i \cap A_j = \emptyset$ for all distinct agents $i, j \in \agents$.

\paragraph{Fairness Notions.} We are interested in finding allocations that are fair, i.e., each agent is satisfied with the items they receive. Envy-freeness is the central fairness notion in fair division.
 \begin{definition}[Envy-freeness]
     An allocation $\alloc = (A_1, \ldots, A_n)$ is \emph{envy-free (EF)} if each agent prefers their own bundle over the bundle of any other agent, i.e., for all agents $i,j \in \agents$, we have $\vtrue_i(A_i) \geq \vtrue_i(A_j)$.
 \end{definition}

Notably, envy-free allocations do not always exist when the items being allocated are indivisible. Hence, our goal in such scenarios is to minimize the maximum envy between any pair of agents. For any two agents $i, j \in \agents$, the envy that agent $i$ has for agent $j$ is defined as $\vtrue_i(A_j) - \vtrue_i(A_i)$. Our objective will be to minimize the maximum envy, i.e., $\max_{i \neq j} \vtrue_i(A_j) - \vtrue_i(A_i)$. In an envy-free allocation, the maximum envy is non-positive.

\paragraph{Inaccurate Preference Model.} Our paper focuses on the setting where we don't get to observe the true values $\vtrue_{i,j}$ of the agents. Instead, we assume that for each agent $i \in [n]$ and item $j \in [m]$, we only have access to a noisy estimate $\vestimate_{i,j}$ of the true values $\vtrue_{i,j}$. Indeed, the envy between the agents is still based on the true values $\vtrue_{i,j}$.

\section{Stochastic True Valuations}\label{subsec: stochastic valuations}

In this section, we focus on the setting where the true valuations of all the agents, $\{\vtrue_{i,j}\}_{i,j}$, are stochastic. \Cref{section:stochastic-val-stochastic-err} describes our result for the case when the noise ($\vestimate_{i,j}-\vtrue_{i,j}$) is stochastic and \Cref{section:stochastic-val-worst-case-errors} describes the result for worst-case, bounded noise.

\subsection{Stochastic Values and Stochastic Noise}\label{section:stochastic-val-stochastic-err}

Let $\mathcal{D}$ and $\mathcal{D}'$ be two real-valued distributions. We consider a model where for every agent $i \in [n]$, item $j\in [m]$, we have $\vtrue_{i,j} \sim \mathcal{D}$, $\eta_{i,j} \sim \mathcal{D}'$, and the algorithm gets to observe values $\vestimate_{i,j} = \vtrue_{i,j} + \eta_{i,j}$. We further assume that the distributions $\mathcal{D}$ and $\mathcal{D}'$ are unknown to the algorithm.
In this setting, we prove that maximizing welfare with respect to the observed values --- a simple and computationally efficient algorithm --- outputs an envy-free allocation, with high probability (\Cref{theorem:observed-value-maximization}). 

\Cref{theorem:observed-value-maximization},  only requires $\mathcal{D}$ to be a continuous distribution with positive variance; $\mathcal{D}'$ can be an arbitrary (discrete or continuous) distribution. It is easy to see that positive variance for $\mathcal{D}$ is a necessary requirement for envy-freeness. If $\mathcal{D}$ has a zero variance, for e.g., $\mathcal{D}$ takes the value $1$ with probability $1$ and $m$ is not a multiple of $n$, then an envy-free allocation does not exist. Continuity of $\mathcal{D}$ is assumed for the ease of exposition; in particular, to avoid tie-breaking in the algorithm.

To prove this result, we first develop a correlation inequality (\Cref{lemma:strict-fkg}) which is a refinement of the well-known Chebyshev's association inequality stated below in~\Cref{lemma:fkg}.

\begin{lemma}[Chebyshev's Association Inequality]\label{lemma:fkg}{\cite{conc_ineq_book}}
    Given any two real-valued, nondecreasing functions $f$, $g$ and a real-valued random variable $X$, $\E[f(X)\cdot g(X)] \geq \E[f(X)] \cdot \E[g(X)]$.
\end{lemma}

The following correlation inequality recovers \Cref{lemma:fkg} and additionally presents the necessary and sufficient condition for a strict inequality to hold. Notably, a strict inequality is crucially needed for the main theorem, \Cref{theorem:observed-value-maximization}. The result holds for arbitrary (discrete, absolutely continuous, singular continuous) random variables; the proof appears in~\Cref{appendix:stochastic-val-stochastic-err}.

\begin{restatable}{theorem}{StrictCorrelationInequality}\label{lemma:strict-fkg} Given any two real-valued, nondecreasing functions $f$, $g$ and a real-valued random variable $X$, we have $\E[f(X)\cdot g(X)] > \E[f(X)] \cdot \E[g(X)]$ if and only if $\var(f(X)) > 0$ and $\var(g(X)) > 0$.
\end{restatable}

\Cref{lemma:strict-fkg} is used to prove \Cref{lemma:noisy-order-statistic}, which plays a central role in \Cref{theorem:observed-value-maximization}. At a high-level, \Cref{lemma:noisy-order-statistic} shows that performing welfare maximization on inaccurate values is, in fact, effective in maximizing welfare with respect to the true valuations; this eventually leads to envy-freeness.

\begin{restatable}{lemma}{noisyValuesCorrelation}\label{lemma:noisy-order-statistic}
%%% AM: I edited this statement to something we might want to have in the final version
    Let $\mathcal{D}$ and $\mathcal{D}'$ be two probability distributions. % whose pdfs are square-integrable. 
    Furthermore, suppose $X_1, \ldots, X_n$, $Y_1,\ldots, Y_n$ are independent random variables such that each $X_i \sim \mathcal{D}$ and each $Y_i \sim \mathcal{D}'$. Then, we have:
    \begin{enumerate}[leftmargin=*]
        \item $\mathbb{E}[X_1 \mid X_1 + Y_1 > \max_{i=2}^n X_i + Y_i] \geq \mathbb{E}[X_1 \mid X_1 + Y_1 \leq \max_{i=2}^n X_i + Y_i]$ if $\mathcal{D}'$ has a positive variance.
        \item $\mathbb{E}[X_1 \mid X_1 + Y_1 > \max_{i=2}^n X_i + Y_i] > \mathbb{E}[X_1 \mid X_1 + Y_1 \leq \max_{i=2}^n X_i + Y_i]$ if $\mathcal{D}$ has a positive variance.%\paritosh{It seems that there should be no condition on $\mathcal{D}'$, i.e., $(2)$ should hold for singular continuous distributions $\mathcal{D}'$ as well.}
        %doesn't have point masses, 
        %then $\mathbb{E}[X_1 \mid X_1 + Y_1 > \max_{i=2}^n X_i + Y_i] > \mathbb{E}[X_1 \mid X_1 + Y_1 \leq \max_{i=2}^n X_i + Y_i]$.
    \end{enumerate}
\end{restatable}
\begin{proof}%%% AM: Try this proof for part (1), possibly part (2)
Define $\mathcal{E}_>$ to be the event $[X_1 + Y_1 > \max_{i=2}^n X_i + Y_i]$, and $\mathcal{E}_\leq$ to be the event $[X_1 + Y_1 \leq \max_{i=2}^n X_i + Y_i]$. 

\paragraph{(1)} If $\mathcal{D}'$ has a positive variance, then the random variables $X_i + Y_i$ have a positive variance. Positive variance implies that $\P(\mathcal{E}_>), \P(\mathcal{E}_\leq) > 0$, i.e., the conditional expectations $\E[X_1 \mid \mathcal{E}_>],\E[X_1 \mid \mathcal{E}_\leq]$ are well defined. 
%First, argue that $\mathbb{P}(\mathcal{E}_>) \in (0,1)$. 
To prove $(1)$, it suffices to show that
$$ \mathbb{E}[X_1 \mid \mathcal{E}_>] \geq \mathbb{E}[X_1] = \mathbb{E}[X_1 \mid \mathcal{E}_>] \mathbb{P}(\mathcal{E}_>) + \mathbb{E}[X_1 \mid \mathcal{E}_{\leq}] \P(\mathcal{E}_{\leq}) , $$
or equivalently, 
$$ \mathbb{E}[X_1 \mid \mathcal{E}_>] = \frac{\mathbb{E}[X_1 \I_{\mathcal{E}_>}]}{\P(\mathcal{E}_>)} \geq \mathbb{E}[X_1] . $$
This can be written as
$$ \mathbb{E}[X_1 \I_{\mathcal{E}_>}] \geq \mathbb{E}[X_1] \mathbb{E}[\I_{\mathcal{E}_>}] . $$
Define $h : \mathbb{R} \rightarrow \mathbb{R}$, $h(x) = x$ and $g : \mathbb{R}^2 \rightarrow \{0,1\}$, $g(x,z) = \I\{x > z\}$. Additionally, let random variable $Z \coloneqq \max_{i=2}^n \left(X_i + (Y_i - Y_1)\right)$; note that $\I_{\mathcal{E}_>} = \I[X_1 > Z]$. Then, we want to show that 
$$ \mathbb{E}[h(X_1) g(X_1,Z)] \geq \mathbb{E}[h(X_1)] \mathbb{E}[g(X_1,Z)] . $$
To prove this, we condition on $Z$. Then, using the association inequality (\Cref{lemma:fkg}) we have that
$$ \mathbb{E}[h(X_1) g(X_1,Z)\mid Z] \geq \mathbb{E}[h(X_1)\mid Z] \mathbb{E}[g(X_1,Z)\mid Z] = \mathbb{E}[h(X_1)] \mathbb{E}[g(X_1,Z)\mid Z] , $$
where we remove the conditioning in $\E[h(X_1) \mid Z]$ due to independence. Then, taking expectations with respect to $Z$ gives us,
$$ \mathbb{E}[h(X_1) g(X_1,Z)] \geq \mathbb{E}[h(X_1)] \mathbb{E}[g(X_1,Z)] . $$

\paragraph{(2)} Based on the previous part, in order to prove $(2)$, we need to establish $\E[h(X_1) g(X_1,Z)] > \E[X_1]\E[g(X_1,Z)]$. Similar to the above, we condition on a fixed value on $Z$. \Cref{lemma:strict-fkg} implies that conditioned on $Z$, the strict inequality
\begin{align}
\mathbb{E}[h(X_1) g(X_1,Z)\mid Z] > \mathbb{E}[h(X_1)] \mathbb{E}[g(X_1,Z)\mid Z],
\end{align}
holds if and only if $\var(h(X_1)) = \var(X_1) > 0$ and $\var(\I\{X_1 > Z\}) = \P[X_1 > Z] (1 - \P[X_1 > Z]) > 0$; here, the former condition is true as per our assumption on $\mathcal{D}$, and the latter condition is equivalent to $\P[X_1 > Z] \in (0,1)$. Hence, for a fixed value of $Z$, if $\P[X_1 > Z] \in (0,1)$, then, and only then, we have $\mathbb{E}[h(X_1) g(X_1,Z)\mid Z] > \mathbb{E}[h(X_1)] \mathbb{E}[g(X_1,Z)\mid Z]$. The final step of the proof will be to show that the strict inequality persists after taking expectations with respect to $Z$. To argue this, we will show that the probability 
$$\P_Z[\P_{X_1}[X_1 > Z] \in (0,1)] > 0,$$
i.e., with positive probability over selecting $Z$, we get a strict inequality conditioned on $Z$, thereby ensuring that the final inequality $\E[h(X_1) g(X_1,Z)] > \E[X_1]\E[g(X_1,Z)]$ obtained after taking expectations with respect to $Z$ is strict as well. % Hence, to complete the proof, we will now show that $\P_Z[\P_{X_1}[X_1 > Z] \in (0,1)] > 0$. Recall that, $Z = \max_{i=2}^n \left(X_i + (Y_i - Y_1)\right)$.

Let $\mathcal{E}$ be the event that $Y_1 \ge \max_{i=2}^n Y_2$. Since $Y_i$s are iid drawn from $\mathcal{D}'$, we have $\P[\mathcal{E}] = \P[Y_1 \ge \max_{i=2}^n Y_i] \ge \frac{1}{n} > 0$. Additionally, conditioned on $\mathcal{E}$, we have $Z \le \max_{i=2}^n X_i$. Hence,
\begin{align*}
    \P_Z[\P_{X_1}[X_1 > Z] \in (0,1)] &\ge \P[\mathcal{E}] \cdot \P_Z[\P_{X_1}[X_1 > Z] \in (0,1) \mid \mathcal{E}] \\
    &\ge \P[\mathcal{E}] \cdot \P_Z[\P_{X_1}[X_1 > \max_{i=2}^n X_i] \in (0,1) \mid \mathcal{E}] \\
    &= \P[\mathcal{E}] \cdot \P_{X_2, \dots, X_n}[\P_{X_1}[X_1 > \max_{i=2}^n X_i] \in (0,1)],
\end{align*}

where we drop the conditioning in the second term because $X_i$s and $Y_i$s are independent and $\mathcal{E}$ only depends on $Y_i$s. Finally, note $\P[\mathcal{E}] > 0$, and the second term in the product above is also positive since the variance of $\mathcal{D}$ is positive. This implies that $\P_Z[\P_{X_1}[X_1 > Z] \in (0,1)]$ is positive and completes the proof.
\end{proof}

We now state the main result of this section, ~\Cref{theorem:observed-value-maximization}; its proof is deferred to~\Cref{appendix:stochastic-val-stochastic-err}.

\begin{restatable}{theorem}{observedValueMaximization}\label{theorem:observed-value-maximization}
Let $\mathcal{D}$ and $\mathcal{D}'$ be two fixed distributions such that $\mathcal{D}$ is continuous and has a positive variance ($\mathcal{D}'$ could be arbitrary). Assume that $\vtrue_{i,j} \sim \mathcal{D}$ for all $i \in [n]$, $j \in [m]$, and we observe the noisy estimates $\vestimate_{i,j} = \vtrue_{i,j} + \eta_{i,j}$ where $\eta_{i,j}$ for all $i \in [n]$, $j \in [m]$. Then, running the welfare maximization algorithm on the observed noisy values $\{\vestimate_{i,j}\}_{i,j}$, computes an envy-free allocation with probability at least $1 - O(1/m)$ if $m/\log{m} \in \Omega(n)$.
\end{restatable}

\subsection{Stochastic Values and Worst-case  Noise}\label{section:stochastic-val-worst-case-errors}
We now consider a general setting where the value vector for each item $j$, $\vtrue_j = (\vtrue_{1, j}, \ldots, \vtrue_{n, j})$ is drawn i.i.d.\@ from a bounded vector-valued distribution $\D$ over $[0, b]^n$ for some bound $b$. This framework allows the agents' valuations for a specific item to be arbitrarily correlated (as the components of a single draw from $\D$). However, values across items are independent.

We show that if $\D$ satisfies a mild condition --- essentially that no two agents are ``identical'' up to scaling --- and if the number of items $m$ is sufficiently large, an efficient algorithm can use noisy estimates to find an envy-free allocation with high probability. The condition is the following: Let $X = (X_1, \ldots, X_N) \sim \D$ be a random vector. We require that $\Pr[X_i = c \cdot X_j] = 0$ for all $i \ne j$ and constants $c \in \mathbb{R}$. 

\begin{restatable}{theorem}{stochasticValuesWorstNoise}\label{theorem:stochastic-value-3-worst-case-noise}
Let $\D$ be a bounded nonnegative distribution satisfying the non-identical condition. 
Assume that $\vtrue_{j} \sim \D$, independently, for all items $j$.
For each $i\in \agents$, $j \in \items$, let $\vestimate_{i,j}$ be a noisy estimate of $\vtrue_{i,j}$ such that for some $\varepsilon > 0$, we have $\max_{i,j}|\vtrue_{i,j} - \vestimate_{i,j}| \leq \varepsilon$. There is a constant $c_\D$ depending only on $\D$ and an algorithm such that if $\varepsilon \leq c_\D$, given the noisy estimates $\vestimate_{i,j}$ for each $i,j$, the algorithm outputs an allocation $\alloc = (A_1, \ldots, A_n)$ such that $\alloc$ is envy-free with probability at least $1 - \exp(-\Omega_\D(m))$.
\end{restatable}

The algorithm uses linear programming to find a strongly envy-free fractional allocation with respect to the noisy estimates $\vestimate$. It then applies randomized rounding to this fractional allocation to obtain the final integral allocation. The size of the linear program is polynomial in $n$ and $m$, with coefficients based on $\vestimate$. Thus, under the mild assumption that the estimates are polynomial in bit length (which must be the case if they were produced by a polynomial time algorithm), the theorem's algorithm runs in polynomial time.

The proof of \Cref{theorem:stochastic-value-3-worst-case-noise} (which appears in \Cref{appendix:stochastic-worst}) relies on establishing the existence of a strongly envy-free allocation in the underlying noiseless setting, which we formalize in the following lemma. This ``strong'' guarantee provides a buffer that is robust to the two sources of error: the bounded noise $\varepsilon$ in the estimates and the error introduced by randomized rounding. This ensures that the final allocation is envy-free with high probability. This existence result may be of independent interest.

\begin{restatable}{lemma}{BarbanelMeasureTheoryLemma}\label{lem: barbanel shit}
Let $\D$ be a distribution satisfying the above conditions.
Assume that $\vtrue_{j} \sim \D$, independently, for all items $j$. Then, there is a constant $c'_\D$ depending only on $\D$ such that with probability at least $1 - \exp(-\Omega_\D(m))$, there exists an allocation $\alloc^*$ such that, for all agents $i,j \in \agents$, $\vtrue_{i}(A^*_i) \geq \vtrue_{i}(A^*_j) + c'_\D \, m$.
\end{restatable}

The proof of \Cref{lem: barbanel shit} relies on the following theorem about the existence of strongly envy-free allocations in cake cutting. Recall, given a $\sigma$-algebra $W$ over a set $C$, a function $\mu: W \to \mathbb{R}_\ge 0$ is a \emph{non-atomic probability measure} if (a) $\mu_i(A) \ge 0$ for every $A \in W$, (b) $\mu_i(\emptyset) = 0$, (c) given a countable collection of pairwise disjoint elements of $W$, $A_1, A_2, \ldots$, $\mu_i(\bigcup_j A_j) = \sum_{j = 1}^\infty \mu_i(A_j)$, (d) if $\mu_i(A) > 0$, then for some $B \subseteq A$ with $B \in W$, $0 < \mu_i(B) < \mu_i(A)$, and (e) $\mu_i(C) = 1$.
\begin{theorem}[Theorem 5.5 (b) of \citet{barbanel2005geometry}]\label{thm:barbanel-original}
    Fix $n$ non-atomic probability measures $\mu_1, \ldots, \mu_n$ over a set $C$ with $\sigma$-algebra $W$. Suppose that no two measures are identical. Then, there exists a partition $(P_1, \ldots, P_n)$ of $C$ such that $\mu_i(P_i) > \mu_i(P_j)$ for all $i \ne j$.
\end{theorem}
We are now ready to prove the lemma.
\begin{proof}[Proof of \Cref{lem: barbanel shit}]

Fix such a distribution $\D$ satisfying the conditions. Let $X = (X_1, \ldots, X_n) \sim \D$ be a random vector. 

The key to this lemma is to represent $\mathcal{D}$ as $n$ different valuations over the unit interval $[0,1]$ and to find a \emph{strongly envy-free allocation} of $[0,1]$. This allocation will correspond to a strongly envy-free randomized allocation. By taking a sufficiently large number of samples, we can convert this into a strongly envy-free integral allocation using concentration inequalities.

Recall that $X_i$ is bounded in $[0, b]$ for some bound $b$.

By Lemma 4.22 of \cite{kallenberg1997foundations}, there exists a measurable function $f: [0,1] \to [0, b]$ such that $f(U)$ has distribution $\D$ when $U \sim \text{Unif}[0,1]$. 
For each $i$, let $f_i$ denote the $i$-th coordinate function of $f$, i.e., $f_i: [0,1] \to[0, b]$. Define the following $n$ set functions $\mu_1, \dots, \mu_n$ on the Borel sigma-algebra $\mathcal{B}([0,1])$:
\[
\mu_i(A) = \int_A f_i \, d\lambda, \quad \text{for all } A \in \mathcal{B}([0,1]),
\]
where $\lambda$ is the Lebesgue measure. Since $\E[|X_i|]$ is finite, each $f_i$ is integrable, and $\mu_i$ is a measure on $[0, 1]$. It is nonnegative as $f_i$ is nonnegative. Since $f_i$ is the Radon-Nikodym derivative of $\mu_i$ with respect to Lebesgue measure, each $\mu_i$ is absolutely continuous and hence nonatomic. We have that $\mu_i(\emptyset) = 0$, and it satisfies countable additivity by properties of Lebesgue integration.

For each $i$, let $\mu'_i(A) = \mu_i(A)/\mu_i([0, 1])$ be its normalized measure. Note that this is well-defined, as if $\mu_i([0, 1]) = 0$, then $\mu_i(A) = 0$ for all Borel $A$, which implies that $f_i(U) = 0$ with probability $1$. This means that $\Pr[X_i = 0 \cdot X_j] = 1$ for $j \ne i$, contradicting the assumption on $\D$.

The Radon-Nikodym derivative of $\mu'_i$ is $f'_i(x) = f_i(x) / \mu_i([0, 1])$. If $\mu'_i = \mu'_j$, their derivatives must be equal almost everywhere (see Theorem 3.8 of \cite{folland1999real}), i.e., $f'_i = f'_j$ almost everywhere. This implies that $f_i(x) = c \cdot f_j(x)$ for $c = \mu_i([0, 1])/\mu_j([0, 1])$ almost everywhere on $[0, 1]$, which again implies that $\Pr[X_i = c \cdot X_j] = 1$, breaking our condition on $\D$.

Using \Cref{thm:barbanel-original}, there is a partition of $[0, 1]$ among the agents $B_1, \ldots, B_n$ such that $\mu'_i(B_i) > \mu'_i(B_j)$ for $i \ne j$. Note that this immediately implies that $\mu_i(B_i) \ge \mu_i(B_j)$ for all $i \ne j$ as well.

Define:
\[
\varepsilon = \min_{i \neq j} (\mu_i(B_i) - \mu_i(B_j)),
\]
which is positive since it is the minimum over a finite number of strictly positive values.

We now interpret $\mu_1, \dots, \mu_n$ as valuation functions over $[0,1]$. 
We observe:
\[
\mathbb{E}[f_i(U) \cdot \I[U \in B_i]] = \int_{B_i} f_i \, d\lambda = \mu_i(B_i),
\]
and thus for any $i \ne j$:
\[
\mathbb{E}[f_i(U) \cdot (\I[U \in B_i] - \I[U \in B_j])] = \mu_i(B_i) - \mu_i(B_j) \geq \varepsilon.
\]

We now sample $m$ i.i.d.\@ draws $X^1, \dots, X^m \sim \mathcal{D}$ and show that with high probability, we obtain a strongly envy-free allocation.
We can couple these draws with uniform draws by sampling $U^1, \dots, U^m \sim \text{Unif}[0,1]$ i.i.d., and setting
\[
X^j = f(U^j).
\]
We define an allocation $\alloc$ by:
\[
A_i = \{ k \mid U^k \in B_i \}.
\]
Set $c'_\D = \varepsilon / 2$, which so far depends only on the distribution $\D$ and not $m$. We will show that with probability $1 - \exp(-\Omega_\D(m))$, we have:
\[
v_i(A_i) - v_i(A_j) > c'_\D \cdot m, \quad \forall i \neq j.
\]
Fix $i \neq j$. Define:
\[
Y^k_{i, j} = X^k_i \cdot (\I[k \in A_i] - \I[k \in A_j]).
\]
Observe that $v_i(A_i) - v_i(A_j) = \sum_k Y^k_{i, j}$. We would like to show that for all $i \ne j$, $\sum_k Y^k_{i, j} \ge c'_\D \cdot m$. To do so, we bound the complement:
$\Pr[\sum_k Y^k_{i, j} < c'_\D \cdot m]$. Note that $\E[\sum_{k} Y^k_{i, j}] \ge m \cdot \varepsilon = 2 c'_\D \cdot m$. Thus, this is a deviation from the mean of $> c'_\D \cdot m$. Furthermore, each $Y^k_{i, j} \in [-b, b]$, $\mathbb{E}[Y_k] \geq \varepsilon$,  the $Y_k$s are independent. We can therefore apply
\emph{Hoeffding’s inequality} to get:
\[
\Pr\left[ \sum_k Y_k < c \cdot m \right] \leq \exp\left( -\frac{c^2 m}{2b^2} \right).
\]
Applying a union bound over all $n(n-1)$ pairs:
\[
\Pr\left[ \exists i \neq j, v_i(A_i) - v_i(A_j) < c \cdot m \right] \leq n(n-1) \exp\left( -\frac{c^2 m}{2b^2} \right) \in \exp\left( -\Omega_\D(m) \right),
\]
completing the proof.
\end{proof}

The purpose of the non-identical condition is now clearer. If two agents had the same values (up to scaling), then it is impossible that both strongly prefer their own bundles.

%The proof of the Lemma relies on adapting results on the existence of strongly envy free allocations from the cake-cutting literature~\citep{barbanel2005geometry}. The full proof of \Cref{theorem:stochastic-value-3-worst-case-noise} can be found in~\Cref{appendix:stochastic-worst}.

\section{Worst-Case True Valuations}\label{subsec: worst-case valuations}

The focus of this section will be on settings where the true preferences of the agents $\vtrue_{i,j}$ are fixed/worst-case. In \Cref{subsection:worst-case-vals-worst-case-noise}, we detail our analysis of the Round-Robin algorithm. \Cref{subsection:applications-of-Round-Robin} presents our two applications of the Round-Robin bound. Finally, in \Cref{subsection:discrepancy-under-noise}, we describe our results on discrepancy and envy minimization.

\subsection{Worst-Case Values and Worst-case  Noise}\label{subsection:worst-case-vals-worst-case-noise}

We start by showing that, given worst-case true valuations and (bounded) worst-case noise, we can find an allocation with maximum envy at most $m/n$ times the noise (\Cref{thm: worst-case valuations and worst-case error upper bound}); we prove that this is the best bound possible by giving a matching lower bound (\Cref{thm: worst-case valuations and worst-case error lower bound}). %The proof of both of these appear in~\Cref{appendix:fair_division}.

The guarantee of~\Cref{thm: worst-case valuations and worst-case error upper bound} is achieved by the simple Round-Robin process. Here, for each agent $i \in [n]$, we assume the existence of a possibly unknown value $\delta_i \geq 0$, which when added to the estimates $\vestimate_{i,j}$, results in a value $\varepsilon$ close to $\vtrue_{i,j}$, i.e., $|\vtrue_{i,j} - (\vestimate_{i,j} + \delta_i)| \leq \varepsilon$. The theorem holds for any values of $\delta_i \geq 0$, and, in particular, the natural choice of $\delta_i = 0$ for all $i \in [n]$. The reason the envy bound in \Cref{thm: worst-case valuations and worst-case error upper bound} is independent of $\delta_i$ is a consequence of the fact that round robin produces \emph{balanced} allocations, ensuring that every agent gets roughly the same number of items. While it is natural to assume that $\delta_i = 0$ for all $i \in [n]$, \Cref{thm: worst-case valuations and worst-case error upper bound} is stated in terms of the additional parameters $\delta_i$s as the balancedness property will be used in an application we will present in \Cref{subsection:applications-of-Round-Robin}.

\begin{restatable}[Upper bound for worst-case noise]{theorem}{roundRobinWorstValWorstNoise}\label{thm: worst-case valuations and worst-case error upper bound}
For every $i\in \agents$ and $j \in \items$, let  $\vestimate_{i,j}$ be the estimate of true value $\vtrue_{i,j} \in [0,b]$, such that $\max_{i,j}|\vtrue_{i,j} - \vestimate_{i,j} - \delta_i| \leq \varepsilon$, for some possibly unknown constants $\varepsilon > 0$, $\delta_i \geq 0$ for each agent $i \in [n]$. Given $\vestimate_{i,j}$ for each $i,j$, we can efficiently compute an allocation $\alloc = (A_1, \ldots, A_n)$ such that, for every pair of agents $i,i' \in \agents$, we have $\vtrue_i(A_{i'}) - \vtrue_i(A_i) \leq 2\varepsilon \left\lceil \frac{m}{n} \right\rceil + b$.
\end{restatable}
\begin{proof}
    Consider the allocation produced by the Round-Robin algorithm, when executed on the estimated values $\vestimate_{i,j}$. 
    Recall that the Round-Robin algorithm allocates items to agents over a sequence of rounds, according to a fixed, arbitrary order. In each round, the next agent in the order picks an item that maximizes her value ($\vestimate_{i,j}$, in our case), among the set of remaining items (ties can be broken in an arbitrary manner). The allocation $\alloc$ produced by Round-Robin satisfies envy-freeness up to one item (EF1) \emph{with respect to the noisy values}. That is, for any pair of agents $i,i' \in \agents$ with $\vestimate_i(A_{i'}) \neq \emptyset$, we have $\vestimate_i(A_i) \geq \vestimate_i(A_{i'} \setminus \{g\})$  for some $g \in A_{i'}$. Furthermore, the allocation $\alloc$ produced by Round-Robin is \emph{balanced}, i.e., all agents get roughly equal number of items, for each agent $i \in [n]$, $|A_i|$ is either $\lfloor \frac{m}{n} \rfloor$ or $\lceil \frac{m}{n} \rceil$. 

    Given these facts, we will bound the envy of an agent $i \in \agents$ for agent $i' \in \agents$ with respect to the true (unknown) values as follows. If $A_{i'} = \emptyset$, then there is no envy since $\vtrue_i(A_{i'}) - \vtrue_i(A_i) = - \vtrue_i(A_i) \leq 0$. Otherwise, if $A_{i'} \neq \emptyset$, let $g \in A_{i'}$ be an item such that  $\vestimate_i(A_i) \geq \vestimate_i(A_{i'} \setminus \{g\})$. Additionally, note that the inequality $|\vtrue_{i,j} - \vestimate_{i,j} - \delta_i| \leq \varepsilon$ implies that $(i)$ $\vtrue_{i}(A_i) \geq \vestimate_{i}(A_i) + \delta_i |A_i| - \varepsilon |A_i|$ and $(ii)$ $\vtrue_{i}(A_{i'} \setminus \{g\}) \leq \vestimate_{i}(A_{i'} \setminus \{g\}) + \delta_i |A_{i'} \setminus \{g\}| + \varepsilon |A_{i'} \setminus \{g\}|$. Using this we get,
    \begin{align*}
        \vtrue_i(A_{i'}) &- \vtrue_i(A_i) \\
        & \leq \vtrue_i(A_{i'} \setminus \{g\}) - \vtrue_i(A_i) + b \tag{$\vtrue_{i,j} \leq b$ for all $j \in \items$}\\
        & \leq \vestimate_i(A_{i'} \setminus \{g\}) + \delta_i |A_{i'}\setminus \{g\}|  + \varepsilon |A_{i'} \setminus \{g\}| - \vtrue_i(A_i) + b \tag{Via $(i)$}\\
        & \leq \left(\vestimate_i(A_{i'} \setminus \{g\}) - \vestimate_i(A_i) \right)  + \delta_i \left(|A_{i'}\setminus \{g\}| - |A_i| \right) + \varepsilon \left(|A_i| + |A_{i'} \setminus \{g\}|\right) + b\tag{Via $(ii)$}\\
        & \leq \left(\vestimate_i(A_{i'} \setminus \{g\}) - \vestimate_i(A_i) \right) + \varepsilon \left(|A_i| + |A_{i'} \setminus \{g\}|\right) + b \tag{$\delta_i \geq 0$, $|A_i|, |A_{i'}| \in \{\left\lfloor \frac{m}{n} \right\rfloor,\left\lceil \frac{m}{n} \right\rceil\}$}\\
        & \leq  \vestimate_i(A_{i'} \setminus \{g\}) - \vestimate_i(A_i) + 2\varepsilon \left\lceil \frac{m}{n} \right\rceil + b \tag{$|A_i|, |A_{i'}| \leq \left\lceil \frac{m}{n} \right\rceil$}\\
        & \leq  2\varepsilon \left\lceil \frac{m}{n} \right\rceil + b. \tag{$\vestimate_i(A_{i'} \setminus \{g\}) \leq \vestimate_i(A_i)$}
    \end{align*}
    This concludes the proof.
    \end{proof}

% \lowerBoundWorstValWorstNoise*

The next theorem shows that the envy bound achieved by Round-Robin is tight (up to an additive term) for deterministic algorithms.

\begin{restatable}[Lower bound for worst-case noise]{theorem}{lowerBoundWorstValWorstNoise}\label{thm: worst-case valuations and worst-case error lower bound}
For every $i\in \agents$ and $j \in \items$, let $\vestimate_{i,j}$ be an estimate of $\vtrue_{i,j}$, such that $\max_{i,j}|\vtrue_{i,j} - \vestimate_{i,j}| \leq \varepsilon$, for some $\varepsilon > 0$. No deterministic algorithm that takes $\vestimate_{i,j}$s as input, and outputs a complete allocation, can guarantee that the envy between any two agents is strictly less than $2\varepsilon \, \frac{m}{n}$.
\end{restatable}
\begin{proof}
Consider any deterministic algorithm, and let $\alloc = (A_1, \ldots, A_n)$ be the allocation returned by this algorithm when $\vestimate_{i,j} = 1/2$ for all $i \in \agents$ and $j \in \items$. Let $p \in \argmax_{i \in \agents} |A_i|$ and $q \in \argmax_{i \in \agents \setminus \{p\}} |A_i|$ be the agent with the largest and second largest number of items in $\alloc$, respectively. We have that $|A_p| + |A_q| \geq \frac{2m}{n}$. We will construct the true values $\vtrue_{i,j}$s to maximize the envy of agent $q$ for agent $p$. In particular, we set $\vtrue_{q,j} = 1/2 + \varepsilon$ for all $j \in A_p$ and $\vtrue_{q,j} = 1/2 - \varepsilon$ for all $j \in A_q$. This implies that the envy of agent $q$ for agent $p$, $\vtrue_q(A_p) - \vtrue_q(A_q) = |A_p| (1/2 + \varepsilon) - |A_q| (1/2 - \varepsilon) =  \varepsilon (|A_p| + |A_q|) + 1/2 (|A_p| - |A_q|)$. Using the fact that $|A_p| \geq |A_q|$ and $|A_p| + |A_q| \geq \frac{2m}{n}$, we get $\vtrue_q(A_p) - \vtrue_q(A_q) \geq 2\varepsilon \, \frac{m}{n}$. %Hence, no deterministic algorithm can ensure that the maximum envy is always below $2\varepsilon \cdot \frac{m}{n}$.
\end{proof}

\subsection{Applications of \Cref{thm: worst-case valuations and worst-case error upper bound}}\label{subsection:applications-of-Round-Robin}

\subsubsection*{{\bf Estimating Preferences Via a BTL Query Model}}

We now describe a method for obtaining a reasonably accurate estimate of all agents' preferences using the Bradley-Terry-Luce (BTL) estimation model. In this model, each agent is asked pairwise comparison queries of the form  "which item do you like more, $k$ or $\ell$?" and it responds with a binary choice from $\{k, \ell\}$. Following the BTL model, the set of queried item pairs for each agent is determined independently by an Erdős–Rényi random graph with edge probability $p$, where each unordered pair $\{k,\ell\}$ is included independently with probability $p$. For every queried pair, we collect $K$ independent comparison responses (i.e., $K$ repeated queries per observed pair). A detailed formulation of the BTL model is provided in \Cref{appendix:BTL-preliminaries}. Using the collected responses, we run a standard maximum-likelihood estimator for each agent $i \in [n]$, to recover $\tilde{\theta}_i \in \R^m$, an approximation of $\theta^*_i$, which is a centered version of the true preference vector of agent $i$. Let $\Theta^* =[{\theta}^*_1; \dots ; {\theta}^*_n]^\T$ and $\hat{\Theta} =[\tilde{\theta}_1; \dots ; \tilde{\theta}_n]^\T$ denote the concatenated true and estimated preference matrices, then the following lemma provides a bound on the estimation error of each $\theta^*_i$. 

\begin{restatable}[{$\ell^2$ and $\ell^\infty$-Error Bounds \cite[Theorem 3.1]{chenetalpaper2}}]{lemma}{chenEtAlSingleAgentBounds}\label{thm:Single Agent Bounds}
 
Assume $p \geq c_0 \frac{\log m}{m}$ for a sufficiently large constant $c_0 > 0$, and that $K$ comparisons are collected per observed pair under an Erdős--Rényi observation graph with parameter $p$. Let $\theta^* \in \R^m$ satisfies $\|\theta^*\|_\infty \leq b$ and $\theta^{*^\T} \1_m = 0$. Then, there exists a polynomial-time estimator that computes $\tilde{\theta} \in \mathbb{R}^{m}$ by asking ${m \choose 2} p K$ average number of pairwise comparisons and satisfies the following bounds
\begin{align*}
\|\tilde{\theta} - \theta^*\|_2 \leq c_2 \cdot \frac{1}{\sqrt{p K}}, &&
\|\tilde{\theta} - \theta^*\|_\infty \leq c_\infty \cdot \sqrt{\frac{\log m}{m p K}},
\end{align*}
 with probability at least $1 - O(1/m^{7})$, for some constants $c_2, c_\infty>0$ only depending on $b$.
\end{restatable}

The estimated preferences obtained from \Cref{thm:Single Agent Bounds} are given to the Round-Robin algorithm to compute an allocation $\alloc = (A_1, A_2, \ldots, A_n)$ with a bounded envy. The guarantees of the end-to-end application combining the BTL model with Round-Robin are stated in \Cref{application:combined_round_robin}.

Before stating~\Cref{application:combined_round_robin}, we note that the combination of BTL estimation with Round-Robin is made possible only because of the balancedness property that Round-Robin exhibits. By design, the BTL model is invariant to linear shifts.\footnote{Formally, the pairwise comparison probability in BTL, stated in~\Cref{eq:btl}, depends only on the difference in preference values and not the mean.} Consequently, we cannot recover the average value of all the items, for any agent using BTL. On the other hand, maximum envy between agents does depend on this average value, and cannot be bounded without it. Using Round-Robin addresses this problem by making envy between any pair of agents independent of the average value of all the items --- this is a direct consequence of the balancedness property, which ensures that agents get a roughly equal number of items.\footnote{This independence is evident from the fact that the envy bound in \Cref{thm: worst-case valuations and worst-case error upper bound} does not depend on the constants $\delta_i$.}

\begin{restatable}[Inference and Fair Division]{application}{mainTheoremEstimationAndFairDivision}\label{application:combined_round_robin}
Assume there are $m$ indivisible items and $n$ additive agents with $n < m$, where $\vtrue_{i,j} \in [0,1]$ is the value of agent $i$ for item $j$. Let $V^* \in [0,1]^{n \times m}$ be the valuation matrix where $V^*_{i,j} = \vtrue_{i,j}$. If the users respond to BTL-based pairwise queries, then there is an algorithm that asks ${m \choose 2} pK$ average number of pairwise-comparison queries to each agent $i \in [n]$, and with probability at least $1 -O(1/m^{6})$, produces an allocation $\alloc = (A_1, \ldots, A_n)$ in polynomial time that has bounded maximum envy. Formally, the obtained allocation $\alloc$ satisfies $$\max_{i,j} v_i(A_i) - v_i(A_j) = O\left(\sqrt{\frac{m \log m}{n^2 p K}}\right),$$
where $K$ is the number of queries per-pair and $p$ is the observation graph probability satisfying $p = \Omega\left(\frac{\log m}{m}\right)$. 
\end{restatable}

\begin{proof}[Proof of \Cref{application:combined_round_robin}]
We begin by transforming the true valuation matrix $V^* \in [0,1]^{n \times m}$ into a preference score matrix $\Theta^*$ via a centering operation. Specifically, for each agent $i \in [n]$ and item $j \in [m]$, we define
    \begin{equation} \label{eq:theta-for-fair-div}
        \theta^*_{i,j} = v^*_{i,j} - \frac{1}{m}\sum_{l=1}^{m} v^*_{i,l}
    \end{equation}
    Centering ensures that $\Theta^*\1_m =\0_n$. Furthermore, as $v^*_{i, j}$s lie in $[0, 1]$, we have $\theta^*_{i,j} \in (-1, 1)$.
    
    Next, for every row of $\Theta^*$, that is each vector $\theta^{*\T}_i$, we invoke \cref{thm:Single Agent Bounds} on $\theta^*_i$ with $b=1$. Under the conditions specified in that theorem, each invocation returns an estimate $\tilde{\theta}_i$ satisfying
    \[
    \|\tilde{\theta}_i - \theta^*_i\|_\infty = O\left(\sqrt{\frac{\log m}{m p K}}\right),
    \]
    with probability at least $1 -O(1/m^{7})$. Therefore, by taking $\hat{\Theta} =[\tilde{\theta}_1; \dots ; \tilde{\theta}_n]^\T$, we have
    \begin{equation}\label{eq:estimation-error-bound}\max_{i, j} |\hat{\Theta}_{i,j} - \Theta^*_{i, j}| = O\left(\sqrt{\frac{\log m}{m p K}}\right)
    \end{equation}
    Through a union bound, we get that \Cref{eq:estimation-error-bound} holds with probability at least $1 -O(n/m^{7})$, which is at least $1-O(1/m^6)$ since $n < m$. Finally, we define the estimated valuation matrix as $\hat{V} \coloneqq \hat{\Theta}$, and for each agent $i$, we define $\delta_i \coloneqq \frac{1}{m}\sum_{j=1}^{m} v^*_{i,j} \in [0,1]$, which represents the mean valuation of agent $i$. Using this along with \Cref{eq:theta-for-fair-div} and (\ref{eq:estimation-error-bound}), for any $i \in [n]$ and $j \in [m]$, we get that
    \begin{align*}
    \left| \vestimate_{i,j} - \big(\vtrue_{i,j} - \frac{1}{m}\sum_{l=1}^m \vtrue_{i,j} \big) \right|= \left|\vestimate_{i,j} - \Theta_{i,j}^*\right|= \left|\hat{\Theta}_{i,j} - \Theta_{i,j}^*\right| = 
         O\left(\sqrt{\frac{\log m}{m p K}}\right) .%     .
    \end{align*}
        
This implies that, for each $i \in [n]$ and $j \in [m]$, an estimate $\vestimate_{i,j}$ that satisfies $|\vtrue_{i,j} - \vestimate_{i,j} - \delta_{i}| = |\vestimate_{i,j} - \vtrue_{i,j} + \delta_{i}| = O\left(\sqrt{\frac{\log m}{m p K}}\right) $ for the value of $\delta_i = \frac{1}{m}\sum_{l=1}^m \vtrue_{i,j} \in [0,1]$. Note that while $\vestimate_{i,j}$ is known, $\delta_{i}$ is unknown. Regardless, we can apply \Cref{thm: worst-case valuations and worst-case error upper bound} to compute an allocation $\alloc = (A_1, A_2, \ldots, A_n)$ in polynomial time that has bounded maximum envy with respect to the true values, i.e., $\max_{i,j} \vtrue_i(A_j) - \vtrue_i(A_i) = O\left(\varepsilon \cdot \frac{m}{n}\right) = O\left(\sqrt{\frac{m \log m}{n^2 p K}}\right)$. This concludes the proof.
\end{proof}

\subsubsection*{{\bf Worst-Case Values and Stochastic Noise}}
We consider the following simple model for stochastic noise: $\vestimate_{i,j} = \vtrue_{i,j} + \varepsilon_{i,j}$, for $\varepsilon_{i,j} = \sigma_{i,j} \cdot z_{i,j}$. We assume that $z_{i,j}$s are drawn i.i.d. from a non-negative distribution $\D$, and $\sigma_{i,j}$s are drawn uniformly at random from the set $\{ -1 , 1 \}$. That is, $z_{i,j}$ is the magnitude of the noise, and $\sigma_{i,j}$ denotes whether the noise underestimates or overestimates the value. The next theorem shows that, when $\D$ is not heavy tailed, then we can efficiently compute an allocation with small envy (with respect to the true valuations). 
Formally, our condition on $\D$ is that $\D$ has monotone hazard rate (MHR): $\frac{1-F(x)}{f(x)}$ is a non-increasing function, where $F(x)$ and $f(x)$ are its cumulative distribution function and probability density function, respectively.

\begin{application}\label{application:worst-case-true-val-stochastic-noise}
For every $i\in \agents$ and $j \in \items$, let $z_{i,j}$s be i.i.d. draws from a non-negative MHR distribution $\D$ such that $\E[\D] \leq \frac{n}{m \log(nm)}$, and $\sigma_{i,j}$s be uniformly random draws from the set $\{ -1 , 1 \}$. Let $\vestimate_{i,j} = \vtrue_{i,j} + \varepsilon_{i,j}$, for $\varepsilon_{i,j} = \sigma_{i,j} \cdot z_{i,j}$. Given $\vestimate_{i,j}$ for each $i,j$, we can efficiently compute an allocation $\alloc = (A_1, \ldots, A_n)$ such that, for every pair of agents $i,i' \in \agents$, we have $\vtrue_i(A_{i'}) - \vtrue_i(A_i) \leq 10$, with probability at least $1 - \frac{1}{(nm)^{3/5}}$.
\end{application}

\begin{proof}[Proof of \Cref{application:worst-case-true-val-stochastic-noise}]
We use the following two lemmas about MHR distributions, where $\D_{k:n}$ denotes the distribution of the $k$-th lowest of $n$ i.i.d. samples from $\D$.

\begin{lemma}[\cite{babaioff2017posting}; Lemma 5.3]\label{lem: babaioff}
For any non-negative MHR distribution $\D$ and $m \leq n$, it holds that $\frac{\E[\D_{m:m}]}{\E[\D_{n:n}]} \geq \frac{H_m}{H_n}$, where $H_k$ is the $k$-th harmonic number.
\end{lemma}

\begin{lemma}[\cite{azizzadenesheli2023reward}; Lemma 3]\label{lem: kamyar}
For any non-negative MHR distribution $\D$, and $n \geq 4$, $\Pr[ \D_{n:n} < 2 \, \E[\D_{n:n} ]] \geq 1 - \frac{1}{n^{3/5}}$.
\end{lemma}

Let $\varepsilon_{\max} = \max_{i,j} | \varepsilon_{i,j} |$ be the magnitude of the largest noise in an instance. $\varepsilon_{\max}$ follows the same distribution as $\D_{nm:nm}$. Therefore, letting $\mu_{\max} = \E[\D_{nm:nm}]$, by~\Cref{lem: kamyar}, $\Pr[ \varepsilon_{\max} < 2 \, \mu_{\max} ] \geq 1 - \frac{1}{(nm)^{3/5}}$. However, by~\Cref{lem: babaioff}, $\mu_{\max} = \E[\D_{nm:nm}] \leq H_{nm} \E[\D] \leq \log(nm) \E[\D]$.
Therefore,
\[
\Pr[ \varepsilon_{\max} < 2 \, \log(nm) \, \E[\D] ] \geq 1 - \frac{1}{(nm)^{3/5}}.
\]

Let $\beta = 2 \, \log(nm) \, \E[\D] \leq \frac{2n}{m}$, and condition on the event that $\varepsilon_{\max} < \beta$.
Then, applying~\Cref{thm: worst-case valuations and worst-case error upper bound} for $\delta_i = 0$ for all agents $i$, we can find an allocation $\alloc = (A_1, \dots, A_n)$ such that $\max_{i,i'} \vtrue(A_{i'}) - \vtrue(A_{i}) \leq 2 \beta \left\lceil \frac{m}{n} \right\rceil + 1 \leq 10$.
\end{proof}

\subsection{Discrepancy Minimization with Inaccurate Inputs
}\label{subsection:discrepancy-under-noise}

This section focuses on a variant of the multicolored discrepancy minimization problem. Upper bounds on multicolored discrepancy are derived, as it is known that these bounds directly translate to the problem of envy minimization~\cite{alweiss2021discrepancy,halpern2025online}. We begin with a brief refresher on discrepancy minimization.

In the vector balancing problem, given a set of vectors $u_1, u_2, \ldots, u_m$ satisfying $\|u_i\|_\infty \leq 1$, the goal is to compute signs $s_1, \ldots, s_m \in \{-1,1\}$ to minimize the \emph{discrepancy}, defined as $\|\sum_{i=1}^m s_i u_i\|_\infty$. The multicolored discrepancy minimization problem is a generalization of the vector balancing problem in which we have to assign one color, among a set of $[k] = \{1, 2,\ldots, k\}$ colors, to each vector, to minimize the \emph{multicolored discrepancy}, i.e., $\max_{i,j \in [k]} \|\sum_{v \in S_i} v - \sum_{v \in S_j} v\|_\infty$, where for each $\ell \in [k]$, $S_\ell$ is the set of vectors that are assigned the color $\ell$. 

To see how envy minimization reduces to multicolored discrepancy minimization, consider the following. If we let each agent $i \in [n]$ correspond to a color (i.e., we set $k=n$) and each item $j \in [m]$ correspond to a vector $\vtrue_j = (\vtrue_{1,j}, \vtrue_{2,j}, \ldots, \vtrue_{n,j})/\sqrt{n}$ (ensuring $\|\vtrue_j\|_\infty \leq 1$), then we get an instance of the multicolored discrepancy minimization problem. Moreover, for any allocation $\alloc = (A_1, \ldots, A_n)$, the maximum envy $\max_{i,j \in [n]} v_i(A_j) - v_i(A_i)$ is always at most the multicolored discrepancy $\max_{i,j \in [n]} \|\sum_{v \in S_i} v - \sum_{v \in S_j} v\|_\infty$ of the corresponding problem, where the set $S_i$ contains the vectors that correspond to the bundle of items $A_i$ assigned to agent $i$. Hence, $\max_{i,j \in [n]} v_i(A_j) - v_i(A_i) \leq \max_{i,j \in [n]} \|\sum_{v \in S_i} v - \sum_{v \in S_j} v\|_\infty$ and any upper bound on multicolored discrepancy translates to an upper bound on maximum envy.

We study envy minimization for a model in which we get to observe the inaccurate preferences $(\vtrue_j + \eta_j) \in \mathbb{R}^n$ for each item $j \in [m]$. We then have to allocate this item to an agent, and after the allocation, the true preference vector $\vtrue_j$ is revealed to us. We will assume that $\vtrue_{i,j} \in [-1,1]$ and $\eta_{i,j} \in [-\varepsilon,\varepsilon]$ are both worst-case.\footnote{This is without loss of generality, as the bounds we obtain can be scaled by a factor of $b$ if $\vtrue_{i,j} \in [-b,b]$ and $\eta_{i,j} \in [-\varepsilon b,\varepsilon b]$.} In this model, we give an efficient \emph{online} algorithm that computes an allocation whose maximum envy is upper bounded by $\mathcal{O}\left(\sqrt{n}(\log{\frac{mn}{\delta}} + m\varepsilon)\right)$ with probability at least $1-\delta$. %This result is formally presented below.

\begin{restatable}{theorem}{OnlineEnvyMinimizationBound}\label{theorem:envy-minimization-discrepancy}
    Let $v^*_1 + \eta_1, v^*_2 + \eta_2, \ldots, v^*_m + \eta_m \in \mathbb{R}^n$ be a sequence of items (valuation vectors) that arrive one at a time, where $v^*_{i,j} \in [-1,1], \eta_{i,j} \in [-\varepsilon, \varepsilon]$ for all $i \in [m], j \in [n]$. When each item $v^*_i + \eta_i$ arrives, we need to allocate it to one of the $n$ agents, after which the true valuation vector $v^*_i$ is revealed to us. In this model, there is an efficient online algorithm that computes an allocation $\alloc = (A_1, A_2, \ldots, A_n)$ that guarantees the following maximum envy bound $\max_{i,j \in [n]} v^*_i(A_j) - v^*_i(A_i) \leq \frac{27\sqrt{n}}{C} \log{\frac{5nm}{\delta}} + 6 m\sqrt{n} \varepsilon$, with probability at least $1 - \delta$, for $C > 0$ being a universal constant following~\Cref{definition:subgaussianity-properties}.
   \end{restatable}

To prove \Cref{theorem:envy-minimization-discrepancy}, we first develop an online algorithm for vector balancing problem with noisy input vectors (\Cref{subsection:vector-balancing}). Using this result, we obtain a bound for the corresponding multicolored discrepancy problem (\Cref{subsection:multicolored-discrepancy}). Finally, via the reduction from envy minimization to multicolored discrepancy minimization explained earlier, we directly obtain \Cref{theorem:envy-minimization-discrepancy}.

\subsubsection{\textbf{Vector balancing with noisy input}}\label{subsection:vector-balancing}

We consider the following weighted vector-balancing problem with inaccurate vectors as input. Formally, given a sequence of inaccurate vectors $\vtrue_1 + \eta_1, \vtrue_2 + \eta_2, \ldots, \vtrue_m + \eta_m$ as input, we want to find signs $s_1, s_2, \ldots, s_m \in \{p, p-1\}$ where $p \in [1/2, 2/3]$ is a fixed parameter, to minimize $\|\sum_{i=1}^m s_i \vtrue_j\|_\infty$.
We show that $\|\sum_{i=1}^m s_i \vtrue_j\|_\infty \leq \mathcal{O}\left(\sqrt{n}(\log{\frac{mn}{\delta}} + m\varepsilon)\right)$ with probability at least $1-\delta$ is achieved by a modified version of the online algorithm given by~\cite{alweiss2021discrepancy}, which achieves a similar logarithmic bound on discrepancy, in the absence of noise. Our algorithm (Balance Under Noise) is stated as \Cref{algo:BalanceUnderNoise}.

\begin{algorithm}[ht]
   \caption{\textsc{Balance Under Noise}}\label{algo:BalanceUnderNoise}
   
   \SetKwInOut{Input}{Input}
   \SetKwInOut{Output}{Output}
   \SetKwFor{While}{while}{do}{end}
   \SetKwFor{ForAll}{for all}{do}{end}
   \SetKwIF{If}{ElseIf}{Else}{if}{then}{else if}{else}{end}
   
   \SetAlgoNlRelativeSize{-1}  % slightly smaller line numbers
   \DontPrintSemicolon          % remove automatic semicolons if you prefer

   \Input{ Parameter $p \in [1/2, 2/3]$ and a sequence of vector $v^*_1 + \eta_1, v^*_2 + \eta_2, \ldots, v^*_m + \eta_m$ where $\|v_t\|_\infty \leq 1, \|\eta_t\|_\infty \leq \varepsilon$ for all $t \in [m]$.}
   \Output{ Signs $s_1, s_2, \ldots, s_m \in \{p, p-1\}$.}
   \BlankLine
\hrule
\BlankLine

Initialize $w_1 = \mathbf{0} \in \mathbb{R}^n$ and $s_1, s_2, \ldots, s_m = 0$.\;
Define function $r(x) \coloneqq \min\{1, \max\{0, x\} \}$ for all $x \in \mathbb{R}$.\label{line:definition-ramp}
   \BlankLine
   \For{$t$ $\gets$ $1, 2, \ldots, m$}{
       Set $c_t \gets \frac{9n}{4C} \log{\frac{5nm}{\delta}} + \varepsilon nt$ \label{line:setting-fail-threshold}\; 
       \BlankLine
       % \If{$ \langle w_t , v^*_t + \eta_t \rangle  > 4(1-p)c_t$ or $ \langle w_t , v^*_t + \eta_t \rangle  < -4pc_t$ or $\|w_t\|_\infty > c_t/\sqrt{n}$\label{line:algorithm-fail-if-condition}}{
       %      Fail \label{line:algorithm-fail}
       % }
       Randomly sample $s_t \sim \begin{cases}
        p, & \text{with probability } \ \ r\left((1-p) - \frac{\langle w_t , v^*_t + \eta_t \rangle}{4c_t}\right), \\
        % p-1,  & \text{with probability } \ \ p + \frac{\langle w_t , v^*_t + \eta_t \rangle}{4c_t}.
        p-1,  & \text{with probability } \ \ 1- r\left((1-p) - \frac{\langle w_t , v^*_t + \eta_t \rangle}{4c_t}\right).
       \end{cases}$\label{line:sampling-sign}

       Set $w_{t+1} \gets w_t + s_t v^*_t$ \label{line:update-of-w-t}
   }
   \Return{signs $s_1, s_2, \ldots, s_m$.} \; 
   \end{algorithm}

\Cref{algo:BalanceUnderNoise} maintains a running discrepancy vector, $w_t = \sum_{j=1}^{t-1} s_j \vtrue_j$. When a new inaccurate vector $\vtrue_t + \eta_t$ arrives, it simply samples the sign $s_t$ from the following distribution: $p$ with probability $r((1-p) - \frac{\langle w_t , v^*_t + \eta_t \rangle}{4c_t})$ and $p-1$ otherwise (Line \ref{line:sampling-sign}). The function $r(\cdot)$ defined in Line \ref{line:definition-ramp}, ensures that the probability value is valid, i.e., lies between zero and one. Note that the probability is computed based on the inaccurate vector $\vtrue_t + \eta_t$. Once the sign $s_t$ has been sampled, the true vector $\vtrue_t$ is revealed, and $w_{t+1}$ is updated (Line \ref{line:update-of-w-t}). 

To establish the discrepancy bound, we will show that the high-dimensional random vector $w_t$ behaves like a subgaussian.\footnote{In \Cref{lemma:subgaussianity-of-w-t}, we will show that $w_t$ satisfies the definition of subgaussianity, albeit with an additional multiplicative factor.} We begin by introducing relevant definitions and a useful lemma.

\begin{definition}[\cite{vershynin2018high}]\label{definition:subgaussianity-properties}
    Every $\sigma$-subgaussian random vector $X \in \mathbb{R}^n$ satisfies the following bounds
    \begin{enumerate}
        \item $\E[\exp(\langle X, u \rangle^2/\sigma^2)] \leq 2$ for all $u \in \mathcal{S}^{n-1}$ and\footnote{The set $\mathcal{S}^{n-1} \coloneqq \{x \in \mathbb{R}^n \mid \|x\|_2 = 1\}$ denotes the unit  sphere in $n$ dimension.}
        \item If $\E[X] = 0$, then $\E[\exp(\langle X, \theta \rangle)] \leq \exp(C \sigma^2 \|\theta\|_2^2)$ for all $\theta \in \mathbb{R}^n$, where $C > 0$ is a universal constant, independent of the random vector.
    \end{enumerate}
\end{definition}

\begin{lemma}[Hoeffding's Lemma]\label{lemma:Hoeffding}
    Let $X$ be any random variable that lies in the range $[a,b]$. Then, for any constant $\lambda \in \mathbb{R}$, we have $\E[\exp(\lambda X)] \leq \exp(\lambda \E[X] + \frac{\lambda^2 (b-a)^2}{8})$.    
\end{lemma}

We will now define an event that will play a crucial role in the analysis of \Cref{algo:BalanceUnderNoise}. Let $\mathrm{Fail_t}$ be the event that during the start of iteration $t$, $w_t$ satisfies the following condition,

\begin{align}\label{equation:definition-of-Fail}
\mathrm{Fail_t} \coloneqq \I[\langle w_t , v^*_t + \eta_t \rangle  > 4(1-p)c_t] \lor \I[\langle w_t , v^*_t + \eta_t \rangle  < -4pc_t] \lor \I[\|w_t\|_\infty > c_t/\sqrt{n}]
\end{align}

That is, $\mathrm{Fail_t}$ is true if either $\|w_t\|_\infty$ is large or $(1-p) - \frac{\langle w_t, v^*_t + \eta_t \rangle}{4c_t} \notin [0,1]$. Note that if the latter condition doesn't hold, i.e., $(1-p) - \frac{\langle w_t, v^*_t + \eta_t \rangle}{4c_t} \in [0,1]$, then we have $r\left((1-p) - \frac{\langle w_t, v^*_t + \eta_t \rangle}{4c_t}\right) = (1-p) - \frac{\langle w_t, v^*_t + \eta_t \rangle}{4c_t}$. We additionally define $\mathrm{Fail_t^c}$ be the complement event of $\mathrm{Fail_t}$. The following lemma proves a concentration result for the random vector $w_t$, showing that it satisfies the definition of subgaussianity (\Cref{definition:subgaussianity-properties}), with an extra multiplicative factor.

   \begin{restatable}{lemma}{concentrationOfRandomVector}\label{lemma:subgaussianity-of-w-t}
    For all $t \in [m]$, the random vector $w_t \in \mathbb{R}^n$ satisfies the following properties: 
    \begin{enumerate}
        \item for all $\theta \in \mathbb{R}^n$, we have $\mathop{\E}\limits_{s_1,\ldots,s_{t-1}}[\exp(\langle w_t, \theta \rangle) \mid \cap_{\ell=1}^{t-1} \mathrm{Fail_\ell^c}] \leq \left(1-\frac{\delta}{m}\right)^{1-t}\exp(\sigma_t^2 \|\theta\|_2^2)$ where $\sigma_t^2 = c_t = \frac{9n}{4C} \log{\frac{5nm}{\delta}} + \varepsilon n t$.
        \item $\P[\mathrm{Fail_t^c} \mid \cap_{\ell=1}^{t-1} \mathrm{Fail_\ell^c}] \geq 1 - \frac{\delta}{m}$.
    \end{enumerate}

   \end{restatable}
   \begin{proof}%[Proof of~\Cref{lemma:subgaussianity-of-w-t}]

    We will prove $(1)$ and $(2)$ property by induction on $t$. Since $w_1 = \mathbf{0}$, both $(1)$ and $(2)$ are trivially satisfied for $t=1$. Now assuming that $(1)$ and $(2)$ are true for $w_t$, we now prove that both of these hold for $w_{t+1}$ as well; this will complete the inductive argument. 

    \noindent
    \textbf{Proving property $(1)$.} For any $\theta \in \mathbb{R}^n$ consider,

    \begin{align*}
        & \mathop{\E}\limits_{s_1,\ldots,s_{t}}[\exp(\langle w_{t+1}, \theta \rangle) \mid \cap_{\ell=1}^{t} \mathrm{Fail_\ell^c}] \\
        & = \mathop{\E}\limits_{s_1,\ldots,s_{t}}[\exp(\langle w_t + s_t v^*_t, \theta \rangle) \mid \cap_{\ell=1}^{t} \mathrm{Fail_\ell^c}] \tag{$w_{t+1} = w_t + s_t v^*_t$} \\
        & = \mathop{\E}\limits_{s_1,\ldots,s_{t-1}}\left[\exp(\langle w_t, \theta \rangle)  \cdot \E_{s_t}[\exp(s_t \langle v^*_t, \theta \rangle) \mid \cap_{\ell=1}^{t} \mathrm{Fail_\ell^c}] \mid \cap_{\ell=1}^{t} \mathrm{Fail_\ell^c}\right] \label{equation:induction-step-subgaussianity-1}\numberthis
    \end{align*}
    To proceed further, we will now invoke Hoeffding's lemma (\Cref{lemma:Hoeffding}) by setting $X = s_t$, $\lambda = \langle v^*_t, \theta \rangle$, $b = p$, and $a = p-1$. This gives us 
    
    \begin{align*}
        \E_{s_t}[\exp(s_t \langle v^*_t, \theta \rangle) \mid \cap_{\ell=1}^{t} \mathrm{Fail_\ell^c}] & \leq \exp\left(\langle v^*_t, \theta \rangle \E[s_t] + \frac{\langle v^*_t, \theta \rangle^2 (p - (p-1))^2}{8}\right) \\
        & = \exp\left( \frac{-\langle v^*_t, \theta \rangle \langle w_t , v^*_t + \eta_t \rangle}{4c_t} + \frac{\langle v^*_t, \theta \rangle^2 }{8}\right), \label{equation:induction-step-subgaussianity-2} \numberthis 
    \end{align*}
    where the final equality uses the fact that given the event $\mathrm{Fail_t^c}$, we have $r\left((1-p) - \frac{\langle w_t, v^*_t + \eta_t \rangle}{4c_t}\right) = (1-p) - \frac{\langle w_t, v^*_t + \eta_t \rangle}{4c_t}$, and therefore, $E[s_t] = -\frac{\langle w_t , v^*_t + \eta_t \rangle}{4c_t}$ (Line \ref{line:sampling-sign}). Substituting \Cref{equation:induction-step-subgaussianity-2} in \Cref{equation:induction-step-subgaussianity-1} gives us,

    \begin{align*}
        & \mathop{\E}\limits_{s_1,\ldots,s_{t}}[\exp(\langle w_{t+1}, \theta \rangle) \mid \cap_{\ell=1}^{t} \mathrm{Fail_\ell^c}] \\
        &\leq  \mathop{\E}\limits_{s_1,\ldots,s_{t-1}}\left[\exp(\langle w_t, \theta \rangle)  \cdot \exp\left( \frac{-\langle v^*_t, \theta \rangle \langle w_t , v^*_t + \eta_t \rangle}{4c_t} + \frac{\langle v^*_t, \theta \rangle^2 }{8}\right) \Big| \cap_{\ell=1}^{t} \mathrm{Fail_\ell^c} \right] \\
         &= \mathop{\E}\limits_{s_1,\ldots,s_{t-1}}\left[\exp\left(\langle w_t, \theta \rangle - \frac{\langle v^*_t, \theta \rangle \langle w_t , v^*_t + \eta_t \rangle}{4c_t}\right) \Big| \cap_{\ell=1}^{t} \mathrm{Fail_\ell^c}  \right] \exp\left(\frac{\langle v^*_t, \theta \rangle^2 }{8}\right)\\
         &= \mathop{\E}\limits_{s_1,\ldots,s_{t-1}}\left[\exp\left(\left\langle w_t, \theta - \frac{\langle v^*_t, \theta \rangle (v^*_t + \eta_t )}{4c_t} \right\rangle \right) \Big| \cap_{\ell=1}^{t} \mathrm{Fail_\ell^c} \right] \exp\left(\frac{\langle v^*_t, \theta \rangle^2 }{8}\right) \label{equation:induction-step-subgaussianity-3} \numberthis 
    \end{align*}
%%%%%% AM Comment: Using the inductive hypothesis below requires us to change the conditioning on the Fail events. This needs to be verified carefully. The rest of the proof works.
    Note that our inductive hypothesis, in particular property $(1)$ for $w_t$ does not apply to \Cref{equation:induction-step-subgaussianity-3} since the event being conditioned on is $\cap_{\ell=1}^{t} \mathrm{Fail_\ell^c}$ and not $\cap_{\ell=1}^{t-1} \mathrm{Fail_\ell^c}$. Nevertheless, using the inductive hypothesis, we can prove the following useful claim whose proof appears in \Cref{app: missing from section 4}.

    \begin{restatable}{claim}{subGaussianityWithDifferentEvent}\label{claim:IH-and-total-probability}
        For all $\beta \in \mathbb{R}^n$, we have $\mathop{\E}\limits_{s_1,\ldots,s_{t-1}}[\exp(\langle w_t, \beta \rangle) \mid \cap_{\ell=1}^{t} \mathrm{Fail_\ell^c}] \leq \left(1-\frac{\delta}{m}\right)^{-t}\exp(\sigma_t^2 \|\beta\|_2^2)$.
    \end{restatable}

    By setting $\beta = \theta - \frac{\langle v^*_t, \theta \rangle (v^*_t + \eta_t )}{4c_t}$ in Claim \ref{claim:IH-and-total-probability} and substituting in \Cref{equation:induction-step-subgaussianity-3} we get
    
        \begin{align*}
        & \left(1-\frac{\delta}{m}\right)^{t} \mathop{\E}\limits_{s_1,\ldots,s_{t}}[\exp(\langle w_{t+1}, \theta \rangle)  \mid \cap_{\ell=1}^{t} \mathrm{Fail_\ell^c} ]\\
        \leq & \exp\left( \sigma_t^2 \cdot \left\|\theta - \frac{\langle v^*_t, \theta \rangle (v^*_t + \eta_t )}{4c_t}\right\|_2^2 \right) \exp\left(\frac{\langle v^*_t, \theta \rangle^2 }{8}\right) \\
        = & \exp\left( \sigma_t^2 \cdot \left( \left\|\theta\right\|_2^2 +  \frac{\langle v^*_t, \theta \rangle^2}{16 c_t^2} \left\|v^*_t + \eta_t\right\|_2^2- \frac{\langle v^*_t, \theta \rangle (\langle v^*_t, \theta \rangle + \langle \eta_t, \theta \rangle)}{2c_t} \right) \right) \exp\left(\frac{\langle v^*_t, \theta \rangle^2 }{8}\right) \\
        = & \exp(\sigma_t^2 \|\theta\|_2^2) \, \exp\left(\langle v^*_t, \theta \rangle^2 \left( \frac{\sigma_t^2}{16 c_t^2}(\|v^*_t\|_2^2 + \|\eta_t\|_2^2 + 2 \langle v^*_t, \eta_t\rangle) - \frac{\sigma_t^2}{2c_t} + \frac{1}{8}\right) \right) \, \exp\left(\frac{-\sigma_t^2 \langle v^*_t, \theta \rangle \langle \eta_t, \theta \rangle}{2c_t}\right).
    \end{align*}
    To simplify the above expression, we will use the fact that $\langle a,b \rangle \leq \|a\|_2 \|b\|_2$, $\|v^*_t\|_2 \leq \sqrt{n}$, $\|\eta_2\|_2 \leq \varepsilon \sqrt{n} \leq \sqrt{n}$, along with $\sigma_t^2 = c_t$. Substituting these gives us,
    \begin{align*}
        &\mathop{\E}\limits_{s_1,\ldots,s_{t}}[\exp(\langle w_{t+1}, \theta \rangle)  \mid \cap_{\ell=1}^{t} \mathrm{Fail_\ell^c}]\\
        & \leq \left(1-\frac{\delta}{m}\right)^{-t} \exp(\sigma_t^2 \|\theta\|_2^2) \cdot \exp\left(\langle v^*_t, \theta \rangle^2 \left( \frac{4n}{16 c_t} - \frac{1}{2} + \frac{1}{8}\right) \right) \cdot \exp\left(\frac{-\langle v^*_t, \theta \rangle \langle \eta_t, \theta \rangle}{2}\right).
    \end{align*}
    Note that for a small value of $\delta$, the term $\frac{4n}{16 c_t} = 1/(\frac{9}{C} \log{\frac{5nm}{\delta}} + 4\varepsilon t)$ is negligible in comparison to the other terms. Hence, $\frac{4n}{16 c_t} - \frac{1}{2} + \frac{1}{8} < 0$ and we get

        \begin{align*}
        \mathop{\E}\limits_{s_1,\ldots,s_{t}}[\exp(\langle w_{t+1}, \theta \rangle) \mid \cap_{\ell=1}^{t} \mathrm{Fail_\ell^c}] & \leq \left(1-\frac{\delta}{m}\right)^{-t}\exp(\sigma_t^2 \|\theta\|_2^2) \exp\left(\frac{-\langle v^*_t, \theta \rangle \langle \eta_t, \theta \rangle}{2}\right) \\
        & \leq \left(1-\frac{\delta}{m}\right)^{-t}\exp(\sigma_t^2 \|\theta\|_2^2) \exp\left(\frac{\varepsilon n \| \theta \|_2^2 }{2}\right) \\
        & \leq \left(1-\frac{\delta}{m}\right)^{-t}\exp\left((\sigma_t^2 +  \varepsilon n) \|\theta\|_2^2 \right) \\
        & = \left(1-\frac{\delta}{m}\right)^{-t} \exp\left(\sigma_{t+1}^2 \|\theta\|_2^2 \right),
    \end{align*}
    where the last equality follows from the fact that $\sigma_{t+1}^2 = \frac{9n}{4C} \log{\frac{5nm}{\delta}} + \varepsilon n(t+1)$. This establishes property $(1)$ for $w_{t+1}$.

    Under the assumption that $\delta \leq 1/2$, an immediate consequence of property $(1)$, which we will use in proving property $(2)$ is that,

    \begin{align*}
    \mathop{\E}\limits_{s_1,\ldots,s_{t}}[\exp(\langle w_{t+1}, \theta \rangle) \mid \cap_{\ell=1}^{t} \mathrm{Fail_\ell^c}] & \leq \left(1-\frac{\delta}{m}\right)^{-t} \exp\left(\sigma_{t+1}^2 \|\theta\|_2^2 \right)\\
    & \leq \left(1-\frac{\delta t}{m}\right)^{-1} \exp\left(\sigma_{t+1}^2 \|\theta\|_2^2 \right)\\
    & \leq \left(1-\delta \right)^{-1} \exp\left(\sigma_{t+1}^2 \|\theta\|_2^2 \right) \tag{$t \leq m$}\\
    & \leq 2  \exp\left(\sigma_{t+1}^2 \|\theta\|_2^2 \right) \numberthis \label{inequality:almost-subGaussian}
    \end{align*}
    where the second inequality follows from the fact that $(1+x)^a \geq 1+xa$ if $x\geq -1, a \geq 1$, and final inequality follows from $\delta \leq 1/2$.
    
    \noindent
    \textbf{Proving property $(2)$.} We will show that $\Pr[\mathrm{Fail_{t+1}} \mid \cap_{\ell = 1}^{t} \mathrm{Fail_\ell^c}] \leq \delta/m$, which will directly prove property $(2)$ for iteration $t+1$. 
    
    Note that $\mathrm{Fail_{t+1}}$ happens only if in iteration $t+1$, $\langle w_{t+1} , v^*_{t+1} + \eta_{t+1} \rangle > 4(1-p)c_{t+1}$ or $\langle w_{t+1} , v^*_{t+1} + \eta_{t+1} \rangle < -4pc_{t+1}$ or $\|w_{t+1}\|_\infty > c_{t+1}/\sqrt{n}$ (\Cref{equation:definition-of-Fail}). Since $p \in [1/2, 2/3]$, we have $4p > 4(1-p)$. We will utilize this fact to simplify as follows, 
    \begin{align*}
        & \Pr[\mathrm{Fail_{t+1}} \mid \cap_{\ell = 1}^{t} \mathrm{Fail_\ell^c}] \\
        = & \Pr[\langle w_{t+1} , v^*_{t+1} + \eta_{t+1} \rangle > 4(1-p)c_{t+1} \text{ or } \langle w_{t+1} , v^*_{t+1} + \eta_{t+1} \rangle < -4pc_{t+1} \text{ or } \|w_{t+1}\|_\infty > c_{t+1}/\sqrt{n} \mid \cap_{\ell = 1}^{t} \mathrm{Fail_\ell^c}] \\
        \leq & \Pr[ | \langle w_{t+1} , v^*_{t+1} + \eta_{t+1} \rangle | > 4(1-p)c_{t+1} \text{ or } \|w_{t+1}\|_\infty > c_{t+1}/\sqrt{n} \mid \cap_{\ell = 1}^{t} \mathrm{Fail_\ell^c}] \tag{$4p > 4(1-p)$}\\
        = & \Pr[ | \langle w_{t+1} , (v^*_{t+1} + \eta_{t+1})/(4(1-p)\sqrt{n}) \rangle | > c_{t+1}/\sqrt{n} \text{ or } \|w_{t+1}\|_\infty > c_{t+1}/\sqrt{n} \mid \cap_{\ell = 1}^{t} \mathrm{Fail_\ell^c}] \\
        \leq & \sum_{u \in S} \Pr[ | \langle w_{t+1} , u \rangle | > c_{t+1}/\sqrt{n} \mid \cap_{\ell = 1}^{t} \mathrm{Fail_\ell^c}],
    \end{align*}
    where set $S \coloneqq \{(v^*_{t+1} + \eta_{t+1})/(4(1-p)\sqrt{n}), e_1, e_2, \ldots, e_n\}$ and $e_1, \ldots, e_n$ are the basis unit vectors along the axes. Note that $\forall u \in S$, we have $\|u\|_2 \leq 3/2$: for all $i \in [n]$, $\|e_i\|_2 = 1$ and $\|(v^*_{t+1} + \eta_{t+1})/(4(1-p)\sqrt{n})\|_2 \leq 2\sqrt{n}/(4(1-p)\sqrt{n}) \leq 2/(4(1-2/3)) = 3/2$, where the last inequality uses the fact that $p \leq 2/3$. We proceed to further evaluate the above probability expression,

    \begin{align*}
        & \Pr[\mathrm{Fail_{t+1}} \mid \cap_{\ell = 1}^{t} \mathrm{Fail_\ell^c}] \\
        \leq & \sum_{u \in S} \Pr[ | \langle w_{t+1} , u \rangle | > c_{t+1}/\sqrt{n} \mid \cap_{\ell = 1}^{t} \mathrm{Fail_\ell^c}] \\
        \leq & \sum_{u \in S} \Pr\left[  \left\langle w_{t+1} , u/\|u\|_2 \right\rangle^2 > \frac{c_{t+1}^2}{(3/2)^2n} \ \Big| \cap_{\ell = 1}^{t} \mathrm{Fail_\ell^c}\right] \tag{$\|u\|_2 \leq 3/2$, $\forall u \in S$}\\
        = & \sum_{u \in S} \Pr\left[  \frac{\left\langle w_{t+1} , u/\|u\|_2 \right\rangle^2}{\sigma_{t+1}^2/C} > \frac{c_{t+1} C}{(3/2)^2n} \ \Big| \cap_{\ell = 1}^{t} \mathrm{Fail_\ell^c}\right]  \tag{$\sigma_{t+1}^2 = c_{t+1}$, for $C < 1/8$}\\
        = & \sum_{u \in S} \Pr\left[  \exp\left( \frac{\left\langle w_{t+1} , u/\|u\|_2 \right\rangle^2}{\sigma_{t+1}^2/C}\right) > \exp\left( \frac{c_{t+1} C}{(3/2)^2n}\right) \ \Big| \cap_{\ell = 1}^{t} \mathrm{Fail_\ell^c}\right] \\\
        \leq & \sum_{u \in S} \ \E\left[ \exp\left( \frac{\left\langle w_{t+1} , u/\|u\|_2 \right\rangle^2}{\sigma_{t+1}^2/C}\right) \ \Big| \cap_{\ell = 1}^{t} \mathrm{Fail_\ell^c}\right] \Big/ \exp\left( \frac{c_{t+1} C}{(3/2)^2n}\right)\tag{Markov's inequality}
    \end{align*}
    % AM Comment: Proof until here works when Fail events are defined using the adjustment that probailities of s_t are capped at 0 or 1 in the algorithm. This ensures Fail events always exist. The next step uses the result from previous lemma (which we assume is true), that conditioned on the event \cap_{\ell = 1}^{t-1} \mathrm{Fail_\ell^c}, w_t above is sub-Gaussian. So, it suffices to check the previous lemma.

    Now, we will make use of the following lemma, whose proof has been deferred to \Cref{app: missing from section 4}.\footnote{The proof of this lemma closely follows~\cite{vershynin2018high}{ Proposition 2.6.1}, differing only in multiplicative constants. We include the proof for completeness.}

    \begin{restatable}{lemma}{DefinitionsOfSubgaussianity}\label{lemma:almost-subGaussianity}
    For a random vector $w \in \mathbb{R}^n$, if there exists a $\lambda > 0$ such that $\forall \theta \in \mathbb{R}^n$, $\E[\exp\left(\langle w, \theta \rangle\right)] \leq 2 \exp\left(C \lambda^2 \|\theta\|_2^2\right)$, where $C < 1/8$ is a fixed constant, then $\forall u \in \mathcal{S}^{n-1}$, we have $\E[\exp\left( \langle w, u\rangle^2/\lambda^2 \right)] \leq 3$.
\end{restatable}
    
    From \Cref{inequality:almost-subGaussian}, we know that the random vector $w_{t+1}$ satisfies the condition in \Cref{lemma:almost-subGaussianity} with $\lambda^2 = \sigma_{t+1}^2/C$. Hence, we get,
    \[\E\left[  \exp\left( \frac{\left\langle w_{t+1} , u/\|u\|_2 \right\rangle^2}{\sigma_{t+1}^2/C}\right) \ \Big| \cap_{\ell = 1}^{t} \mathrm{Fail_\ell^c}\right] \leq 3.\]
    Using the above bound, and plugging in $c_{t+1} = \frac{9n}{4C} \log{\frac{5nm}{\delta}} + \varepsilon n(t+1)$ in the running chain of inequalities,
    \begin{align*}
         \Pr[\mathrm{Fail_{t+1}} \mid \cap_{\ell = 1}^{t} \mathrm{Fail_\ell^c}] & \leq \sum_{u \in S} 3 \Big/ \exp\left( \log{\frac{5nm}{\delta}} + \frac{4C \varepsilon ({t+1})}{9} \right)\\
         & = (n+1)\cdot \frac{3\delta}{5nm \exp(4C(t+1)\varepsilon/9)} \tag{$|S| = n+1$}\\
         & \leq \frac{\delta}{m} \tag{$4C(t+1)\varepsilon/9)> 0$, $n \geq 2$}
    \end{align*}
    This completes the proof.
   \end{proof}

   \Cref{lemma:subgaussianity-of-w-t} directly leads to our main theorem about vector balancing.

    \begin{theorem}\label{corollary:infinity-norm-bound}
        During the execution of \Cref{algo:BalanceUnderNoise}, the random vector $w_t = \sum_{i=1}^{t-1} s_i v^*_i$ satisfies $\|w_t\|_\infty \leq \frac{9\sqrt{n}}{4C} \log{\frac{5nm}{\delta}} + \varepsilon t \sqrt{n}$, for all $t \in [m]$ with probability at least $1-\delta$.
   \end{theorem}
   \begin{proof}
       We will prove that with probability at least $1 - \delta$, the event $\cup_{\ell=1}^m \mathrm{Fail_t}$ doesn't happen during the execution of \Cref{algo:BalanceUnderNoise}, i.e., \Cref{equation:definition-of-Fail} is never satisfied.
    The theorem then follows, since \Cref{equation:definition-of-Fail} not being satisfied means that the inequality $\|w_t\|_\infty \leq c_t/\sqrt{n} = \frac{9\sqrt{n}}{4C} \log{\frac{5nm}{\delta}} + \varepsilon t \sqrt{n}$ holds for all $t$.

    Using chain rule of probability we have $\Pr[\cap_{t=1}^m \mathrm{Fail_t^c}] = \Pi_{t=1}^m \Pr[\mathrm{Fail_t^c} \mid \cap_{\ell < t}\mathrm{Fail_\ell^c}]$. From \Cref{lemma:subgaussianity-of-w-t}, we know that $\Pr[\mathrm{Fail_t^c} \mid \cap_{\ell < t}\mathrm{Fail_\ell^c}] \geq 1- \delta/m$. Hence, we get $\Pr[\cap_{t=1}^m \mathrm{Fail_t^c}] = \Pi_{t=1}^m \Pr[\mathrm{Fail_t^c} \mid \cap_{\ell < t}\mathrm{Fail_\ell^c}] \geq (1-\delta/m)^m \geq 1-\delta$. This completes the proof.
   \end{proof}

   \subsubsection{\textbf{Multicolored discrepancy and Envy minimization.}}\label{subsection:multicolored-discrepancy}
    We now generalize the result of~\Cref{corollary:infinity-norm-bound} to the multicolored discrepancy problem (\Cref{theorem:multicolored-discrepancy-bound}). The idea is based on a known reduction that has been utilized in prior work~\cite{alweiss2021discrepancy,halpern2025online}.

    We begin by first describing the algorithm that achieves the multicolored discrepancy bound of \Cref{theorem:multicolored-discrepancy-bound}. Consider a complete and full binary tree $\mathcal{T}$ having exactly $k$ leaves and root node $r$.\footnote{A binary tree is \emph{full} if every internal node has either $0$ or $2$ children. A binary tree is \emph{complete} if all levels except possibly the last one are completely filled.} For each internal node $u$ of $\mathcal{T}$, we use $n^u_\ell$ and $n^u_r$ to denote the number of leaves in its left and right subtree, respectively. We run a copy of \Cref{algo:BalanceUnderNoise} at each internal node $u$ of $\mathcal{T}$ by setting the parameter $p$ to $p_u = n^u_\ell/(n^u_\ell + n^u_r)$.

    When each vector $v^*_j + \eta_j$ arrives, we pass it down along the tree from the root to one of the leaves; $v^*_i + \eta_j$ is then assigned a color equal to the index of the leaf it eventually arrives at. In particular, first $v^*_j + \eta_j$ is passed through the copy of \Cref{algo:BalanceUnderNoise} running at the root $r$ of $\mathcal{T}$. If $v^*_j + \eta_j$ is assigned a sign of $p_r$ by \Cref{algo:BalanceUnderNoise} running at $r$, then it is passed down to the next node in the right subtree, if $v^*_j + \eta_j$ is assigned a sign of $p_r - 1$, then it is passed down to the next node in the left subtree. Supposed the next node is $u$, then we again check whether $v^*_j + \eta_j$ is assigned a sign of $p_u,$ $p_u - 1$ at $u$, and pass it down to the right, left subtree respectively. This continues till the vector $v^*_j + \eta_j$ reaches a leaf node, i.e, it gets assigned a color equal to the index of the leaf node. When this happens, the true vector $v^*_j$ is revealed. Using $v^*_j$, the value of the vector $w^u_j$ (superscript $u$ identifies the node) used in \Cref{algo:BalanceUnderNoise} is updated at each internal node through which $v^*_j + \eta_j$ passes; this update happens in Line \ref{line:update-of-w-t}. This process of assigning vectors to leaves is repeated for all the vectors. This completes the description of the multicolored discrepancy algorithm. 
    
    The guarantee of this algorithm is stated in \Cref{theorem:multicolored-discrepancy-bound}, its proof appears in \Cref{app: missing from section 4}. Notably, while we run \Cref{algo:BalanceUnderNoise} in every internal node of $\mathcal{T}$ having a depth of tree $O(\log{k})$, the final multicolored discrepancy bound we obtain is only a constant factor larger than \Cref{corollary:infinity-norm-bound}.

   \begin{restatable}{theorem}{multicoloredDiscrepancy}\label{theorem:multicolored-discrepancy-bound}
    Let $k\geq 2$ be an integer and $v^*_1 + \eta_1, v^*_2 + \eta_2, \ldots, v^*_m + \eta_m \in \mathbb{R}^n$ be a sequence of vectors that arrive one at a time, where $\|v^*_i\|_\infty \leq 1, \|\eta_i\|_\infty \leq \varepsilon \leq 1$ for all $i \in [m]$. When each vector $v^*_i + \eta_i$ arrives, we need to assign it a color from $\{1, 2, \ldots, k\}$ after which the vector $v^*_i$ is revealed to us. In this model, there is an efficient algorithm that achieves a multicolored discrepancy bound of $\max_{i,j \in [k]} \| \sum_{v \in S_i} v - \sum_{v \in S_j} v \|_\infty \leq \frac{27\sqrt{n}}{2C} \log{\frac{5nmk}{\delta}} + 6\varepsilon t\sqrt{n}$, with probability at least $1 - \delta$, where $S_a$ is the set of vectors that are assigned the color $a \in [k]$ after all the vectors have arrived.
   \end{restatable}
   
   Finally, we obtain \Cref{theorem:envy-minimization-discrepancy} by setting $k=n$ in \Cref{theorem:multicolored-discrepancy-bound} and recalling that any upper bound on multicolored discrepancy is also an upper bound for the corresponding envy-minimization problem, as highlighted at the start of \Cref{subsection:discrepancy-under-noise}.

\OnlineEnvyMinimizationBound*

\section*{Acknowledgments}

Anuran Makur was supported in part by the NSF CAREER Award under Grant CCF-2337808. Alexandros Psomas and Paritosh Verma were supported in part by an NSF CAREER Award under Grant CCF-2144208, and a research award from the Herbert Simon Family Foundation. Paritosh Verma also wants to acknowledge the support of the Bilsland Dissertation Fellowship. Trung Dang was funded in part by NSF award CCF-2217069.

% In the interest of anonymization, please do not include acknowledgements in your submission.
%
%\begin{acks}
%
%	The authors would like to thank Dr. Maura Turolla of Telecom
%	Italia for providing specifications about the application scenario.
%
%	The work is supported by the \grantsponsor{GS501100001809}{National
%		Natural Science Foundation of
%		China}{http://dx.doi.org/10.13039/501100001809} under Grant
%	No.:~\grantnum{GS501100001809}{61273304\_a}
%	and~\grantnum[http://www.nnsf.cn/youngscientsts]{GS501100001809}{Young
%		Scientsts' Support Program}.
%
%
%\end{acks}

% Bibliography

\bibliographystyle{ACM-Reference-Format}
\bibliography{references}

% Appendix
\appendix
\section{Preliminaries for the BTL model}\label{appendix:BTL-preliminaries}
This appendix provides the formal setup for agents' preference elicitation via the BTL model. The material here expounds on the estimation procedure described in \Cref{subsection:applications-of-Round-Robin}.

%\paragraph{Multi-agent BTL Model.}  
We consider a natural extension of the well-studied BTL model \cite{BradleyTerry1952, Plackett1975, Luce1959}, where each agent has a separate set of BTL preferences. Specifically, for $n$ agents and $m$ items, the model assumes that each agent $i \in [n]$ has an associated (unknown) preference score $\theta^*_{i,j} \in [-b, b]$ for every item $j \in [m]$. Under the BTL model, the probability of agent $i$ preferring item $j$ over item $k$ depends only on the difference in their latent scores and is given by:
     \begin{equation} \label{eq:btl}
     \mathbb{P}(\text{agent } i \text{ prefers item } j \text{ over } k ) = \frac{\exp(\theta^*_{i,j})}{\exp(\theta^*_{i,j}) + \exp(\theta^*_{i,k})} = \psi(\theta_{i,j}^* - \theta_{i,k}^*),
     \end{equation}   
    where $\psi(t) = 1/(1+\exp(-t))$ denotes the sigmoid function.

    We collect the preference scores as a matrix $\Theta^* \in \mathbb{R}^{n \times m}$ such that $\Theta^*_{i,j} = \theta^*_{i,j}$. Our goal is to estimate the unknown preference matrix $\Theta^*$ from noisy pairwise comparison queries collected from the agents. Since the preference probabilities in \cref{eq:btl} are invariant to translation, we impose the following identifiability constraint 
     \begin{equation}\label{eq:theta} 
        \Theta^* \mathbf{1}_m = \mathbf{0}_n.
    \end{equation}     % \paragraph{Observation Graphs.}  
    
     For each agent, comparisons are observed only on a subset of item pairs. This comparison structure is modeled by an agent-specific undirected graph $\mathcal{E}_i$. Specifically, for each agent $i \in \agents $, an edge $(j,k) \in \mathcal{E}_i$ if and only if agent $i$ provides a comparison between items $j$ and $k$. We assume that the observation graph of each agent is drawn independently from an Erdős-Rényi random graph model with parameter $p$. % \paragraph{Pairwise Comparisons.}  
     
     We assume that for every agent $i \in [n]$ and for each pair $ (j,k) \in \mathcal{E}_i$ with $k > j$,  we observe $K$ independent comparisons. Let $y^{(i,l)}_{j,k} \in \{0,1\}$ denote the outcome of the $l$-th comparison between items $j$ and $k$ by agent $i$, where $y^{(i,l)}_{j,k} = 1$ indicates a preference for item $j$ over item $k$. Under the BTL model, these outcomes are generated according to
     \begin{equation}\label{eq:yjkil definition}
         y_{j,k}^{(i,l)} =
         \begin{cases}
         1, & \text{ with probability } \sF(\theta_{i,j}^*- \theta_{i,k}^*) \\
         0, & \text{ otherwise}
         \end{cases} ,
     \end{equation}
     For convenience, we set $y_{k,j}^{(i,l)} = 1 - y_{j,k}^{(i,l)}$ for all $(k>j)$ with $(j,k) \in \mathcal{E}_{i}$. We define the empirical preference probabilities $y_{j,k}^{(i)} = \frac{1}{K} \sum_{l=1}^K y_{j,k}^{(i,l)}$. 
     Given the observed empirical preference probabilities $\{y_{j,k}^{(i)}\}_{i,j,k} $, our goal is to estimate the latent preference matrix $ \Theta^*$. 
     
For each agent $i \in [n]$, we estimate the agent's parameters $\tilde{\theta}^{(i)}$ by solving the following marginal log-likelihood problem using only agent $i$'s preference data (cf. \cite{chenetalpaper2}). This problem takes the form
\begin{equation}\label{initial regularized likelihood}
    \tilde\theta^{(i)}= \argmin_{\theta: \theta^\T\1_m =0} \left\{-\frac{1}{mp}\sum_{ (j,k) \in \mathcal{E}_i } y_{j,k}^{(i)}\log\psi(\theta_j-\theta_k)  +(1-y_{j,k}^{(i)})\log(1-\psi(\theta_j-\theta_k))  \right\}.
\end{equation}
This convex optimization yields $n$ vectors $\tilde{\theta}^{(i)} \in \R^m$, one for each agent. Stacking these vectors as rows produces the initial matrix estimate $\hat{\Theta} = \left[\tilde{\theta}^{(1) \T}; \tilde{\theta}^{(2) \T}; \ldots; \tilde{\theta}^{(n) \T} \right] \in \mathbb{R}^{n \times m}.$
The estimation error bounds for this estimator are quantified in \Cref{thm:Single Agent Bounds}.

\newcommand{\Alg}{\textsc{Alg}}
\newcommand{\envy}{\operatorname{envy}}
\newcommand{\Unif}{\operatorname{Unif}}
\newcommand{\Bern}{\operatorname{Bern}}
\newcommand{\Binom}{\operatorname{Binom}}

\section{Missing proofs from~\Cref{section:stochastic-val-stochastic-err}}\label{appendix:stochastic-val-stochastic-err}

\StrictCorrelationInequality*
\begin{proof}
    Note that if $\var(g(X)) = 0$ or $\var(f(X)) = 0$, then, almost surely, $g(X) = \E[g(X)]$ or $f(X) = \E[f(X)]$ respectively. Both of these, in turn, imply the equality $E[f(X)\cdot g(X)] = E[f(X)] \cdot E[g(X)]$. Hence, the condition $\var(f(X)) > 0$ and $\var(g(X)) > 0$ is necessary.

    To prove its sufficiency, we consider the difference $\E[f(X)\cdot g(X)] - E[f(X)] \cdot E[g(X)]$, which, can be equivalently stated as $E\left[f(X)\cdot g(X) - f(X) \cdot \E[g(X)]\right] = E\left[f(X)\left( g(X) - \E[g(X)]\right)\right]$. In particular, we will show that if $\var(f(X)) > 0$ and $\var(g(X)) > 0$, then $E\left[f(X)\left( g(X) - \E[g(X)]\right)\right] > 0$ as well. This will complete the proof.

    Define $\mathcal{E}_>$ and $\mathcal{E}_<$ to respectively denote the events $g(X) > \E[g(X)]$ and $g(X) < \E[g(X)]$. Since $\var(g(X)) > 0$, both these events can happen, i.e., have a positive probability. To proceed, we will use the observation that the random variable $g(X) - \E[g(X)]$ can be written as $\I_{\mathcal{E}_>} (g(X) - \E[g(X)]) - \I_{\mathcal{E}_<} (\E[g(X)] - g(X))$. Hence, we want to show that the difference stated below is positive,
    $$E\left[f(X)\left( g(X) - \E[g(X)]\right)\right] = \E[f(X) \cdot \I_{\mathcal{E}_>} (g(X) - \E[g(X)])] - \E[f(X)\cdot \I_{\mathcal{E}_<} (\E[g(X)] - g(X))].$$

    Note that $0 = \E[g(X) - \E[g(X)]] = \E[\I_{\mathcal{E}_>} (g(X) - \E[g(X)])] - \E[\I_{\mathcal{E}_<} (\E[g(X)] - g(X))]]$. Therefore, we have $\E[\I_{\mathcal{E}_>} (g(X) - \E[g(X)])] = \E[\I_{\mathcal{E}_<} (\E[g(X)] - g(X))]]$. The fact that $\var(g(X)) > 0$, further implies that, 
    $$C \coloneqq \E[\I_{\mathcal{E}_>} (g(X) - \E[g(X)])] = \E[\I_{\mathcal{E}_<} (\E[g(X)] - g(X))]] > 0.$$

    To prove that $\E[f(X) \cdot \I_{\mathcal{E}_>} (g(X) - \E[g(X)])] - \E[f(X)\cdot \I_{\mathcal{E}_<} (\E[g(X)] - g(X))] > 0$, we will instead prove that $\E[f(X) \cdot \I_{\mathcal{E}_>} (g(X) - \E[g(X)])]/C - \E[f(X)\cdot \I_{\mathcal{E}_<} (\E[g(X)] - g(X))]/C > 0$. Indeed, the latter implies the former. Now, using the definition of $C$, we can write,
    \begin{align*}
        & \E[f(X) \cdot \I_{\mathcal{E}_>} (g(X) - \E[g(X)])]/C - \E[f(X)\cdot \I_{\mathcal{E}_<} (\E[g(X)] - g(X))]/C \\
        & = \frac{\E[f(X) \cdot \I_{\mathcal{E}_>} (g(X) - \E[g(X)])]}{\E[\I_{\mathcal{E}_>} (g(X) - \E[g(X)])]} - \frac{\E[f(X)\cdot \I_{\mathcal{E}_<} (\E[g(X)] - g(X))]}{\E[\I_{\mathcal{E}_<} (\E[g(X)] - g(X))]]}.
    \end{align*}

    In the above expression, define $u \coloneqq \frac{\E[f(X) \cdot \I_{\mathcal{E}_>} (g(X) - \E[g(X)])]}{\E[\I_{\mathcal{E}_>} (g(X) - \E[g(X)])]}$ and $\ell \coloneqq \frac{\E[f(X)\cdot \I_{\mathcal{E}_<} (\E[g(X)] - g(X))]}{\E[\I_{\mathcal{E}_<} (\E[g(X)] - g(X))]]}$ to be the left and right terms, respectively. Additionally, let $x_u \coloneqq \inf\{x : g(X) > \E[g(X)]\}$ and $x_\ell \coloneqq \sup\{x : g(X) < \E[g(X)]\}$. The following two observations will be crucial for the subsequent proof, where we will prove $u > \ell$.

    \begin{enumerate}
        \item [(i)] $f(x_u) \geq f(x_\ell)$. Since $g$ is nondecreasing, we have $x_\ell \leq x_u$, which, along with the fact that $f$ is nondecreasing implies that $f(x_u) \geq f(x_\ell)$.
        
        \item [(ii)] $u \geq f(x_u)$ and $f(x_\ell) \geq \ell$. Note that $u$, by definition, is a weighted mean (or a convex combination) of $f(X)$, for all values of $X$ that satisfy $g(X) > \E[g(X)]$, i.e., all values of $X$ that are at least $x_u$. Hence, we have $u \geq f(x_u)$. Similarly, $\ell$ is a weighted mean of $f(X)$ for values of $X$ that satisfy $E[g(X)] > g(X)$, i.e., values that are at most $x_\ell$. This gives us the inequality $f(x_\ell) \geq \ell$.
    \end{enumerate}
    
    Note that together $(i)$ and $(ii)$ imply,

    \begin{equation}\label{equation:strict-correlation-key-eq}
        u \geq f(x_u) \geq f(x_\ell) \geq \ell,
    \end{equation}
    
    i.e., we get the weak inequality $u \geq l$. However, we want to show that $u - \ell > 0$. To establish the strict inequality, we will show that at least one of the following three strict inequalities must be true: $f(x_u) > f(x_\ell)$, $u > f(x_u)$, or $f(x_\ell) > \ell$. Either of these three strict inequalities along with \Cref{equation:strict-correlation-key-eq} will imply that $u - \ell > 0$ and will complete the proof.

    Indeed, if we have $f(x_u) > f(x_\ell)$, then we're done. Henceforth, we focus on the case when $f(x_u) = f(x_\ell)$. Given that $f(x_u) = f(x_\ell)$ and $\var(f(X)) > 0$, it must be the case that $\var(f(X) \mid \mathcal{E}_>) > 0$ or $\var(f(X) \mid \mathcal{E}_<) > 0$ --- this follows from the law of total variance.

    We will show that this implies that either $u > f(x_u)$ or $f(x_\ell) > \ell$. This will rely on the following crucial property. The weighted mean $u$, of $f(X)$, for values of $X$ satisfying $g(X) > \E[g(X)]$, has the property that each value of $X = \hat{x}$ (where $g(\hat{x}) > \E[g(X)]$) has a strictly positive weight of $\frac{g(\hat{x}) - \E[g(X)]}{\E[\I_{\mathcal{E}_>} (g(X) - \E[g(X)])]}$ in the weighted mean. Hence, if we have $\var(f(X) \mid \mathcal{E}_>) > 0$, then we must have $u > f(x_u)$. Similarly, the weighted mean $\ell$, of $f(X)$, for values of $X$ satisfying $g(X) < \E[g(X)]$, is such that each $X = \tilde{x}$ with $\E[g(\tilde{x})] > g(\tilde{x})$ has a strictly positive weight of $\frac{\E[g(X)] - g(\tilde{x})}{\E[\I_{\mathcal{E}_<} (\E[g(X)] - g(X))]}$ in the weighted mean $\ell$. Therefore $u > \ell$: either $f(x_u) > f(x_\ell)$ is true, else $\var(f(X) \mid \mathcal{E}_<) > 0$ or $\var(f(X) \mid \mathcal{E}_>) > 0$, which implies that either $u > f(x_u)$ or $f(x_\ell) > \ell$. The equivalence is proved.
\end{proof}

\observedValueMaximization*
\begin{proof}
To allocate the items to the agents, we run the social welfare maximization algorithm on the noisy estimates, i.e., we allocate each item $j \in [m]$ to the agent $i \in \argmax_{i \in [n]}\{\vestimate_{i,j}\}$ (tie-breaking randomly). We observe that this algorithm satisfies the following two properties: first, for each item $j \in [m]$, we have $\mathbb{P}(i \in \argmax_{i \in [n]}\{\vestimate_{i,j}\}) = 1/n$, i.e., every agent has an equal, $1/n$, chance of getting each item --- this follows from the fact that $\vestimate_{i,j}$ are continuous random variables, and are independent and identically distributed. 

Second, we have that $\mathbb{E}[\vtrue_{i,j} \mid i \in \argmax_{k \in [n]}\{\vestimate_{k,j}\}] > \mathbb{E}[\vtrue_{i,j} \mid i \notin \argmax_{k \in [n]}\{\vestimate_{k,j}\}]$, i.e., agents have a higher \emph{true} value for an item they receive based on the \emph{observed} noisy values. This is the most crucial property of our analysis, and it follows from \Cref{lemma:noisy-order-statistic} since $\mathcal{D}$ is continuous and has a positive variance. With these two properties, we will show that this algorithm results in envy-free allocations with high probability if $m/\log{m} = \Omega(n)$.\footnote{The subsequent argument will follow a similar analysis presented in \cite{dickerson2014computational}.} 

Define indicator random variables $X_{i,j} \coloneqq \mathbbm{1}[i \in \argmax_{k \in [n]}\{\vestimate_{k,j}\}]$, which are true if and only if agent $i$ received item $j$. Define $\mathcal{B}_{i,\ell} = \mathbbm{1}[\sum_{j = 1}^m v_{i,j} X_{i,j} < \sum_{j = 1}^m v_{i,j} X_{\ell,j}]$ to be the event that agent $i$ envies agent $\ell$ after the execution of the algorithm. We want to argue that the probability $\P(\cup_{i \neq \ell} \mathcal{B}_{i,\ell})$ is tiny. Towards this, we begin by analyzing the $\P[\mathcal{B}_{i,\ell}]$. 

For all $j \in [m]$, define $\mu_j \coloneqq \mathbb{E}[v_{i,j} \mid i \in \argmax_{k \in [n]}\{\vestimate_{k,j}\}]$ and $\mu_j' \coloneqq \mathbb{E}[v_{i,j} \mid i \notin \argmax_{k \in [n]}\{\vestimate_{k,j}\}]$. Note that as items behave identically, we have $\mu_1 = \mu_2 = \ldots = \mu_m$ and $\mu'_1 = \mu'_2 = \ldots = \mu'_m$, so we define $\mu = \mu_1$ and $\mu' = \mu'_1$. Furthermore, for fixed $\mathcal{D}$ and $\mathcal{D}'$, we know that $\mu > \mu'$ 
from part (2) of \Cref{lemma:noisy-order-statistic}, since $\mathcal{D}$ and $\mathcal{D}'$ satisfy the required conditions.

Note that for the event $\mathcal{B}_{i,\ell}$ to happen, either the event $\mathcal{B}^1_{i,\ell} \coloneqq \mathbbm{1}[\sum_{j = 1}^m v_{i,j} X_{i,j} < \frac{m}{n}\left(\frac{\mu + \mu'}{2}\right)]$ or $\mathcal{B}^2_{i,\ell} \coloneqq \mathbbm{1}[\sum_{j = 1}^m v_{i,j} X_{\ell,j} > \frac{m}{n}\left(\frac{\mu + \mu'}{2}\right)]$ must happen; indeed, if $\mathcal{B}^1_{i,\ell}$ and $\mathcal{B}^2_{i,\ell}$ do not happen, then $\mathcal{B}_{i,\ell}$ cannot happen. Hence, to show that $\mathcal{B}_{i,\ell}$ happens with small probability, we will show that $\mathcal{B}^1_{i,\ell}$ and $\mathcal{B}^2_{i,\ell}$ happen with small probability.

Note that $\E[\sum_{j = 1}^m v_{i,j} X_{i,j}] = m \cdot \P(X_{i,j} = 1) \cdot \mathbb{E}[v_{i,j} \mid i \in \argmax_{k \in [n]}\{\vestimate_{k,j}\}] = \frac{m\mu}{n}$ and $\E[\sum_{j = 1}^m v_{i,j} X_{\ell,j}] = m \cdot \P(X_{\ell,j} = 1) \cdot \mathbb{E}[v_{i,j} \mid i \in \argmax_{k \in [n]}\{\vestimate_{k,j}\}] = \frac{m\mu'}{n}$. Using this, we can rewrite
\[\P(\mathcal{B}^1_{i,\ell}) = \P\left(\sum_{j = 1}^m v_{i,j} X_{i,j} < \E[\sum_{j = 1}^m v_{i,j} X_{i,j}] \left(1 - \frac{\mu - \mu'}{2\mu}\right)\right)\] and \[\P(\mathcal{B}^2_{i,\ell}) = \P\left(\sum_{j = 1}^m v_{i,j} X_{\ell,j} > \E[\sum_{j = 1}^m v_{i,j} X_{\ell,j}]\left(1 + \frac{\mu - \mu'}{2\mu'}\right)\right).\]

Define $\delta = \min\{\frac{\mu - \mu'}{2\mu} ,1\}$, and note this is a fixed quantity, as $\mu$ and $\mu'$ are fixed. Since $\sum_{j=1}^m v_{i,j}X_{i,j}$ is a nonnegative random variable, we get that

\[\P(\mathcal{B}^1_{i,\ell}) = \P\left(\sum_{j = 1}^m v_{i,j} X_{i,j} < \E[\sum_{j = 1}^m v_{i,j} X_{i,j}] \left(1 - \frac{\mu - \mu'}{2\mu}\right)\right) = \P\left(\sum_{j = 1}^m v_{i,j} X_{i,j} < \E[\sum_{j = 1}^m v_{i,j} X_{i,j}] \left(1 - \delta \right)\right).\] Additionally, the fact that $\mu > \mu'$ implies that
\[\P(\mathcal{B}^2_{i,\ell}) = \P\left(\sum_{j = 1}^m v_{i,j} X_{\ell,j} > \E[\sum_{j = 1}^m v_{i,j} X_{\ell,j}] \left(1 + \frac{\mu - \mu'}{2\mu'}\right)\right) \le \P\left(\sum_{j = 1}^m v_{i,j} X_{\ell,j} > \E[\sum_{j = 1}^m v_{i,j} X_{\ell,j}] \left(1 + \delta \right)\right).\] Now, we can use the Chernoff bounds on the sum of independent-and-$[0,1]$-bounded variables $v_{i, j} X_{i, j}$ and $v_{i, j} X_{\ell, j}$ to obtain the following upper bound on $\P(\mathcal{B}_{i,\ell})$,

\begin{align*}
    \P(\mathcal{B}_{i,\ell}) & \leq \P(\mathcal{B}^1_{i,\ell} \cup \mathcal{B}^2_{i,\ell})\\
    & \leq \P(\mathcal{B}^1_{i,\ell}) +  \P(\mathcal{B}^2_{i,\ell}) \tag{union bound}\\
    & \leq \exp\left(-\E\left[\sum_{j = 1}^m v_{i,j} X_{i,j}\right]\frac{\delta^2}{2}\right) + \exp\left(-\E\left[\sum_{j = 1}^m v_{i,j} X_{\ell,j}\right]\frac{\delta^2}{3}\right) \tag{Chernoff bound}\\
    & = \exp\left(-\frac{m\mu}{n}\cdot \frac{\delta^2}{2}\right) + \exp\left(-\frac{m\mu'}{n}\cdot \frac{\delta^2}{3}\right) < 2 \exp\left(-\frac{m\mu'}{n}\cdot \frac{\delta^2}{3}\right).
\end{align*}
    For $m$ where $m \geq \frac{9n \log{m}}{\mu' \delta^2}$, we get $\P(\mathcal{B}_{i,\ell}) \leq 2\exp(-3 \log{m})= 2/m^3$. Hence, via a union bound, we get that the probability of the resulting allocation not being envy-free is $\P(\cup_{i\neq \ell} \mathcal{B}_{i,\ell}) \leq \frac{2n^2}{m^3} \in O(1/m)$. This completes the proof.
\end{proof}

\section{Missing proofs from~\Cref{section:stochastic-val-worst-case-errors}}\label{appendix:stochastic-worst}

% We first prove \Cref{lem: barbanel shit}, restated here for convenience.

% \BarbanelMeasureTheoryLemma*

We now prove the full theorem, again restated here for convenience.
\stochasticValuesWorstNoise*
\begin{proof}[Proof of~\Cref{theorem:stochastic-value-3-worst-case-noise}]

Let $\alloc^* = (\alloc^*_1, \ldots, \alloc^*_n)$ be the allocation that minimizes the maximum envy with respect to the true values, $\max_{i, j \in [n]} \vtrue_i(\alloc^*_j) - \vtrue_i(\alloc^*_i)$. By \Cref{lem: barbanel shit}, this is at most $- c'_\D \cdot m$ with probability $1 - \exp(-\Omega_\D(m))$, where $c'_\D$ is defined in the lemma. Condition on this event. 

Let $\hat{\alloc} = (\hat{A}_1, \dots, \hat{A}_n)$ be the allocation that minimizes the maximum \emph{observed} envy, $\max_{i,j \in \agents} \vestimate_i( A_j ) - \vestimate_i( A_i )$. Since $\alloc^*$ is a feasible allocation, $\beta$ can be no larger than the observed envy gap of $\alloc^*$, which can be at most $\varepsilon \cdot m$ larger than its true gap. Thus, $\beta \le -c'_\D \cdot m + \varepsilon \cdot m$.

We will choose $c_\D = c'_\D / 4$, so $\varepsilon \le c'_\D / 4$. Plugging this in, we get that $\beta \le -(3/4) \cdot c'_\D m$. 

Now consider the following linear program that minimizes the maximum envy (with respect to observed values) over all fractional allocations.
\begin{align*}
    \min \alpha\\
  \alpha \geq \sum_{j=1}^m \vestimate_{i,j} x_{i',j} - \vestimate_{i,j} x_{i,j} && \forall i \in [n], i' \in [n], i' \neq i\\
  \sum_{i=1}^n x_{i,j} = 1 && \forall j \in [n]\\
  x_{i,j} \geq 0 && \forall i \in [n], j \in [m]
\end{align*}
The integral allocation $\hat{\alloc}$ is a feasible solution for this LP, therefore, the value of optimal solution of this LP, $x^*$, has maximum envy at most $\beta$. Let $\mu_{i,i'} = \sum_{j=1}^m x^*_{i',j} \vestimate_{i,j}$. Let $\tilde{\alloc} = (\tilde{A}_1, \dots, \tilde{A}_n)$ be the integral allocation we get by allocating item $j$ to agent $i$ with probability $x^*_{i,j}$, and let $\mu_{i, i'} = \E[ \vestimate_i(\tilde{A}_{i'}) ]$. Note that $\mu_{i,i'} - \mu_{i,i} \leq \beta \le -(3/4)c'_\D m$.

The expected observed value from agent $i$ for bundle $\tilde{A}_{i'}$ is $\E[\vestimate_i(\tilde{A}_{i'})] = \mu_{i, i'}$. The value $\vestimate_i(\tilde{A}_{i'})$ is the sum of $m$ independent random variables (one for each item), each bounded in the range $[0, \vestimate_{i, j}] \subseteq [0, b + \varepsilon]$. Since $\varepsilon\le c_\D$, the range is $[0, b + c_\D]$. By Hoeffding's inequality, for any $\delta > 0$:
\[\Pr[ | \vestimate_i(\tilde{A}_{i'}) - \mu_{i,i'} | \geq \delta ] \leq 2 \exp \left( \frac{-2\delta^2}{m (b + c_\D)^2} \right) .\]

Applying a union bound over all $\le n^2$ pairs, the probability that any observed envy deviates from its mean by more than $\delta$ is \[\Pr[ \forall i,i': | \vestimate_i(\tilde{A}_{i'}) - \mu_{i,i'} | \leq \delta ] \geq 1 - 2 \exp \left( \frac{-4 \delta^2 \ln(n)}{m} \right).\]

We will choose $\delta = m \cdot c'_\D / 4$. This ensures that except with probability $\exp(-\Omega(m))$, $\vestimate_i(\tilde{A}_{i'}) - \vestimate_i(\tilde{A}_i) \le \beta + 2\delta \le -m \cdot c'_\D/4$.

Finally, note that $\vtrue_i(\tilde{A}_{i'}) - \vtrue_i(\tilde{A}_i)$ can deviate from the observed gap by at most $m \cdot \varepsilon \le m \cdot c'_\D /4$. This ensures that $\tilde{A}$ is envy-free. A union bound over the two events yields the theorem.
\end{proof}

\section{Missing proofs from~\Cref{subsection:discrepancy-under-noise}}\label{app: missing from section 4}

\subGaussianityWithDifferentEvent*
    \begin{proof}
        Using the law of total probability we get,

        \begin{align*}  \mathop{\E}\limits_{s_1,\ldots,s_{t-1}}[\exp(\langle w_t, \beta \rangle) \mid \cap_{\ell=1}^{t-1} \mathrm{Fail_\ell^c}] & = \mathop{\E}\limits_{s_1,\ldots,s_{t-1}}[\exp(\langle w_t, \beta \rangle) \mid \cap_{\ell=1}^{t} \mathrm{Fail_\ell^c}]\cdot \P[\mathrm{Fail_t^c} \mid \cap_{\ell=1}^{t-1} \mathrm{Fail_\ell^c}] +\\
        & \ \ \ \ \ \ \ \ \ \ \ \ \  \mathop{\E}\limits_{s_1,\ldots,s_{t-1}}[\exp(\langle w_t, \beta \rangle) \mid \mathrm{Fail_t}, \cap_{\ell=1}^{t-1} \mathrm{Fail_\ell^c}]\cdot \P[\mathrm{Fail_t} \mid \cap_{\ell=1}^{t-1} \mathrm{Fail_\ell^c}] \\
        & \geq \mathop{\E}\limits_{s_1,\ldots,s_{t-1}}[\exp(\langle w_t, \beta \rangle) \mid \cap_{\ell=1}^{t} \mathrm{Fail_\ell^c}]\cdot \P[\mathrm{Fail_t^c} \mid \cap_{\ell=1}^{t-1} \mathrm{Fail_\ell^c}]\\
        & \geq \mathop{\E}\limits_{s_1,\ldots,s_{t-1}}[\exp(\langle w_t, \beta \rangle) \mid \cap_{\ell=1}^{t} \mathrm{Fail_\ell^c}] \left(1-\frac{\delta}{m}\right),
        \end{align*}
        where the final inequality follows from our induction hypothesis, i.e., $\P[\mathrm{Fail_t^c} \mid \cap_{\ell=1}^{t-1} \mathrm{Fail_\ell^c}] \geq 1-\frac{\delta}{m}$. Hence, we get $\mathop{\E}\limits_{s_1,\ldots,s_{t-1}}[\exp(\langle w_t, \beta \rangle) \mid \cap_{\ell=1}^{t} \mathrm{Fail_\ell^c}] \leq \left(1-\frac{\delta}{m}\right)^{-1}\mathop{\E}\limits_{s_1,\ldots,s_{t-1}}[\exp(\langle w_t, \beta \rangle) \mid \cap_{\ell=1}^{t-1} \mathrm{Fail_\ell^c}]$. Substituting the bound $\mathop{\E}\limits_{s_1,\ldots,s_{t-1}}[\exp(\langle w_t, \beta \rangle) \mid \cap_{\ell=1}^{t-1} \mathrm{Fail_\ell^c}] \leq \left(1-\frac{\delta}{m}\right)^{1-t} \exp(\sigma_t^2 \|\beta\|^2_2)$, from the induction hypothesis, completes the proof.
    \end{proof}

\DefinitionsOfSubgaussianity*
\begin{proof}
    Using the Taylor series expansion we have,
    \begin{align}\label{equation:taylor-expansion} 
    \E[\exp\left( \langle w, u\rangle^2/\lambda^2 \right)] = 1+ \sum_{i=1}^\infty \frac{\E[\langle w,u\rangle^{2i}]}{i! \cdot  \lambda^{2i}} 
    \end{align}
    To bound the series above, we will bound the moments $\E[\langle w,u\rangle^{2i}]$ for all even $i \geq 1$. Towards this note that,

    \begin{align*}
        \E[\langle w,u\rangle^{2i}] & = \int_0^\infty  \P[\langle w,u\rangle^{2i} > x] \, dx \\
        & = \int_0^\infty  \P[\exp\left(\langle w, \alpha u\rangle \right) > \exp(\alpha x^{1/2i})] \, dx \tag{for any $\alpha > 0$}\\
        & \leq \int_0^\infty  \frac{\E[\exp\left(\langle w, \alpha u\rangle \right)]}{\exp(\alpha x^{1/2i})} \, dx \tag{Markov's inequality}\\
        & \leq 2 \int_0^\infty  \frac{\exp\left(C \lambda^2 \alpha^2 \right)}{\exp(\alpha x^{1/2i})} \, dx \tag{$\E[\exp\left(\langle w, \alpha u \rangle\right)] \leq 2 \exp\left(C \lambda^2 \alpha^2 \right)$}\\
        & = 2 \int_0^\infty  \exp\left( - \frac{x^{1/i}}{4C\lambda^2} \right) \, dx \tag{setting $\alpha = x^{1/2i}/(2C\lambda^2)$}\\
        & = 2 \int_0^\infty  (4C\lambda^2)^{i} \cdot i \cdot y^{i-1}\exp\left( - y \right) \, dy \tag{substituting $y =\frac{x^{1/i}}{4C\lambda^2}$}\\
        & = 2i (4C\lambda^2)^{i} \int_0^\infty   y^{i-1}\exp\left( - y \right) \, dy \\
        & = 2i (4C\lambda^2)^{i} \cdot \Gamma(i)  \numberthis \label{inequality:moment-bound}
    \end{align*}
    where the final equality follows from the definition of the well-known gamma function $\Gamma(z) \coloneqq \int_0^\infty x^{z-1} e^{-z} \, dz$. Substituting the bound in \Cref{inequality:moment-bound} into \Cref{equation:taylor-expansion}, we get
    \begin{align*}
        \E[\exp\left( \langle w, u\rangle^2/\lambda^2 \right)] & = 1 + \sum_{i=1}^\infty \frac{\E[\langle w,u\rangle^{2i}]}{i! \cdot  \lambda^{2i}} \\
        & \leq 1+ \sum_{i=1}^\infty \frac{2i (4C\lambda^2)^{i} \cdot \Gamma(i)}{i! \cdot  \lambda^{2i}}\\
        & \leq 1 +  2 \sum_{i=1}^\infty (4C)^{i} \tag{$\Gamma(i) = (i-1)!$}\\
        & \leq 1 + 2 \sum_{i=1}^\infty 2^{-i} = 3 \tag{$C < 1/8$}
    \end{align*}
    This completes the proof.
\end{proof}

\multicoloredDiscrepancy*
   \begin{proof}

    Now we will show that this algorithm described in \Cref{subsection:multicolored-discrepancy}, achieves a multicolored discrepancy bound of $\max_{i,j \in [k]} \allowbreak \| \sum_{v \in S_i} v - \sum_{v \in S_j} v \|_\infty \leq \frac{27\sqrt{n}}{2C} \log{\frac{4nmk}{\delta}} + 6\varepsilon t\sqrt{n}$ with probability at least $1-\delta$. Here, for any node $u$, $S_u$ is the set of vectors that pass through node $u$ during the execution of the multicolored discrepancy algorithm. To assist with the analysis, for each node $u$, we define $\pi_u$ to be the fraction of total leaves that are in the subtree rooted at $u$, i.e., the number of leaves in the subtree of $u$ divided by $k$. We now state a key claim.
    \begin{claim}\label{claim:multicolored-bound-induction}
    With probability at least $1-\delta$, for all nodes $u$ of $\mathcal{T}$, we have $\|\sum_{v \in S_u} v - \pi_u \sum_{v \in S_r} v\|_\infty \leq 3  (\frac{9\sqrt{n}}{4C} \log{\frac{4nmk}{\delta}} + \varepsilon m \sqrt{n})$
    \end{claim}

    Before proving this claim, we will show that it suffices to establish the desired multicolored discrepancy bound. Given Claim \ref{claim:multicolored-bound-induction}, for all leaves $i,j \in [k]$, we will have
    \begin{align*}
        \Big\|\sum_{v \in S_i} v - \sum_{v \in S_j} v\Big\|_\infty & = \Big\| \Big( \sum_{v \in S_i} v - \frac{1}{k} \sum_{v \in S_r} v \Big) - \Big(\sum_{v \in S_j} v - \frac{1}{k} \sum_{v \in S_r} v \Big) \Big\|_\infty \\
        & \leq \Big\| \sum_{v \in S_i} v - \pi_i \sum_{v \in S_r} v \Big\|_\infty + \Big\| \sum_{v \in S_j} v - \pi_j \sum_{v \in S_r} v \Big\|_\infty \tag{$\pi_i = \pi_j = 1/k$}\\
        & \overset{\text{(Claim }\ref{claim:multicolored-bound-induction}\text{)}}{\leq} 6 \cdot \left(\frac{9\sqrt{n}}{4C} \log{\frac{4nmk}{\delta}} + \varepsilon m \sqrt{n}\right) = \frac{27\sqrt{n}}{2C} \log{\frac{4nmk}{\delta}} + 6\varepsilon m \sqrt{n}.
    \end{align*}
    The above bound holds for all $i,j \in [k]$ with probability at least $1-\delta$, giving us the desired multicolored discrepancy bound. As the final step, we prove Claim \ref{claim:multicolored-bound-induction}.

    \begin{proof}[Proof of Claim \ref{claim:multicolored-bound-induction}]
    From \Cref{corollary:infinity-norm-bound}, we know that for each internal node $u$ in $\mathcal{T}$, $\|w^u_m\|_\infty \leq \frac{9\sqrt{n}}{4C} \log{\frac{4nmk}{\delta}} + \varepsilon t \sqrt{n}$ with probability at least $1-\delta/k$. Since $\mathcal{T}$ has exactly $k$ leaves, it has $k-1$ internal nodes. Therefore, through a union bound, we get that with probability at least $1-\delta$, all internal nodes $u$ will have $\|w^u_m\|_\infty \leq \frac{9\sqrt{n}}{4C} \log{\frac{4nmk}{\delta}} + \varepsilon m \sqrt{n}$. Additionally, note that if $u_\ell$ and $u_r$ are respectively the left and right children of $u$, then we have $\|w^u_m\|_\infty = \|p_u \sum_{v \in S_{u_r}} v - (1-p_u)\sum_{v \in S_{u_\ell}} v\|_\infty \leq \frac{9\sqrt{n}}{4C} \log{\frac{4nmk}{\delta}} + \varepsilon m \sqrt{n}$.
    
    We will prove the claim by induction, starting from the root node and going down the tree $\mathcal{T}$. For the base case note that for the root node $r$, we have $\|\sum_{v \in S_r} v - \pi_u \sum_{v \in S_r} v\|_\infty = \|\sum_{v \in S_r} v - 1 \cdot \sum_{v \in S_r} v\|_\infty = 0$, i.e., the bound is trivially satisfied. Now, assuming that this bound is satisfied for an internal node $u$, we will show that it will be satisfied for both its left child $u_\ell$ and right child $u_r$. We will prove the argument only for the left child $u_\ell$; the argument for the right child $u_r$ will follow similarly. 
    
    As per our induction hypothesis, we have $\|\sum_{v \in S_u} v - \pi_u \sum_{v \in S_r} v\|_\infty \leq 3  (\frac{9\sqrt{n}}{4C} \log{\frac{4nmk}{\delta}} + \varepsilon m \sqrt{n})$. Additionally, the copy of \Cref{algo:BalanceUnderNoise} running at $a$, guarantees us that $\|w^u_m\|_\infty = \|(1-p_u)\sum_{v \in S_{u_\ell}} v - p_u\sum_{v \in S_{u_r}} v\|_\infty \leq \frac{9\sqrt{n}}{4C} \log{\frac{4nmk}{\delta}} + \varepsilon m \sqrt{n}$; this follows from \Cref{corollary:infinity-norm-bound}. Using this we can bound,

    \begin{align*}
        & \ \|\sum_{v \in S_{u_\ell}} v - \pi_{u_\ell} \sum_{v \in S_{u_r}} v\|_\infty \\
      = & \ \|(p_u + 1- p_u) \sum_{v \in S_{u_\ell}} v - \pi_{u_\ell} \sum_{v \in S_{u_r}} v\|_\infty \\
      = & \ \|(p_u + 1- p_u) \sum_{v \in S_{u_\ell}} v - p_u \pi_u \sum_{v \in S_r} v\|_\infty \tag{$\pi_{u_\ell} = p_u \pi_u$}\\
      = & \ \|p_u(\sum_{v \in S_u} v - \sum_{v \in S_{u_r}} v) + (1- p_u) \sum_{v \in S_{u_\ell}} v - p_u \pi_u \sum_{v \in S_r} v\|_\infty \tag{$\sum_{v \in S_\ell} v = \sum_{v \in S_a} v - \sum_{v \in S_r} v$}\\
      = & \ \|p_u(\sum_{v \in S_u} v - \pi_u \sum_{v \in S_r} v ) + (1- p_u) \sum_{v \in S_{u_\ell}} v - p_u \sum_{v \in S_{u_r}} v \|_\infty \\
      \leq & \ p_u \|\sum_{v \in S_u} v - \pi_u \sum_{v \in S_r} v \|_\infty + \|(1- p_u) \sum_{v \in S_{u_\ell}} v - p_u \sum_{v \in S_{u_r}} v \|_\infty \\
      \leq & \ \frac{2}{3} \|\sum_{v \in S_u} v - \pi_u \sum_{v \in S_r} v \|_\infty + \|(1- p_u) \sum_{v \in S_{u_\ell}} v - p_u \sum_{v \in S_{u_r}} v \|_\infty  \tag{$p_a \in [\frac{1}{2}, \frac{2}{3}]$}\\
      \leq & \ \frac{2}{3} \cdot 3  (\frac{9\sqrt{n}}{4C} \log{\frac{4nmk}{\delta}} + \varepsilon m \sqrt{n}) + \frac{9\sqrt{n}}{4C} \log{\frac{4nmk}{\delta}} + \varepsilon m \sqrt{n} \\
      = & \ 3  (\frac{9\sqrt{n}}{4C} \log{\frac{4nmk}{\delta}}).
    \end{align*}
    This completes the induction step and therefore the proof.
    \end{proof}

   \end{proof}

\end{document}